\documentclass[letterpaper,10pt,journal]{IEEEtran}

\usepackage{graphics}   
\usepackage{epsfig}      
\usepackage{mathptmx}   
\usepackage{times}      
\usepackage{amsmath}    
\usepackage{amssymb}    
\usepackage{amsthm}     
\usepackage{xcolor}
\usepackage{enumerate}
\usepackage[caption=false,font=footnotesize]{subfig}
\usepackage{dsfont}
\usepackage[noadjust]{cite}

\newtheorem{theorem}{Theorem}
\newtheorem{proposition}{Proposition}
\newtheorem{assumption}{Assumption}
\newtheorem{lemma}{Lemma}
\newtheorem{cor}{Corollary}
\newtheorem{remark}{Remark}
\newtheorem*{example}{Example}
\newtheorem{definition}{Definition}

\newcommand{\E}{\mathbb{E}}
\newcommand{\Ep}[1]{\E\big[#1\big]}
\newcommand{\R}{\mathbb R}
\newcommand{\N}{\mathbb N}
\newcommand{\Gcal}{\mathcal G}
\newcommand{\Vcal}{\mathcal V}
\newcommand{\Ecal}{\mathcal E}

\newcommand{\Fab}{\mathcal F_{\alpha,\beta}}
\newcommand{\Sab}{\mathcal S_{a,b}}
\newcommand{\Sao}{\mathcal S_{a,0}}
\newcommand{\Qbf}{\mathbf Q}
\newcommand{\Lpbf}{\mathbf L_p}
\newcommand{\invL}{{L^{\dagger}}}

\newcommand{\prt}[1]{\left(#1\right)}
\newcommand{\brk}[1]{\left[#1\right]}
\newcommand{\brc}[1]{\left\{#1\right\}}

\newcommand{\norm}[1]{\lVert#1\rVert}
\newcommand{\LdagNorm}[1]{\lVert#1\rVert_{\Lpbf^\dag}}
\providecommand{\quotes[1]}{``#1"}
\newcommand{\inProd}[2]{\langle#1,#2\rangle}

\DeclareMathOperator*{\argmin}{arg\,min}
\newcommand{\optx}[1]{x^{\ast,#1}}
\newcommand{\amax}{a_{+}}
\newcommand{\amin}{a_{-}}

\title{
Random Coordinate Descent for Resource Allocation in Open Multi-Agent Systems
}

\author{Charles Monnoyer de Galland, Renato Vizuete, Julien M. Hendrickx, Elena Panteley, and Paolo Frasca 
\thanks{Research supported by the “RevealFlight” ARC at UCLouvain, by the \textit{Incentive Grant for Scientific Research (MIS)} \quotes{Learning from Pairwise Data} of the F.R.S.-FNRS and in part by the Agence Nationale de la Recherche (ANR) via grant “Hybrid And Networked Dynamical sYstems” (HANDY), number ANR-18-CE40-0010.}
\thanks{C.~Monnoyer de Galland and R.~Vizuete equally contributed to this work.   C.~Monnoyer de Galland, R.~Vizuete and J. M. Hendrickx are with the ICTEAM institute, UCLouvain, Louvain-la-Neuve, Belgium. R.~Vizuete was with Universit\'{e} Paris-Saclay, CNRS, CentraleSup\'{e}lec, Laboratoire des signaux et syst\`{e}mes, France, and C.~Monnoyer de Galland was a FRIA fellow (F.R.S.-FNRS) during the elaboration of this work. E.~Panteley is with Universit\'{e} Paris-Saclay, CNRS, CentraleSup\'{e}lec, Laboratoire des signaux et syst\`{e}mes, 91190, Gif-sur-Yvette, France. P.~Frasca is with Univ.\ Grenoble Alpes, CNRS, Inria, Grenoble INP, GIPSA-lab, F-38000 Grenoble, France. (E-mail adresses: ~charles.monnoyer@uclouvain.be; ~renato.vizueteharo@uclouvain.be; ~julien.hendrickx@uclouvain.be;
~elena.panteley@l2s.centralesupelec.fr;~paolo.frasca@gipsa-lab.fr).}}

\begin{document}
\maketitle

\begin{abstract}
We propose a method for analyzing the 
distributed random coordinate descent algorithm for solving separable resource allocation problems in the context of an open multi-agent system, where agents can be replaced during the process. 
In particular, we characterize the evolution of the distance to the minimizer in expectation by following a time-varying optimization approach which builds on two components.
First, we establish the linear convergence of the algorithm in closed systems, in terms of the estimate towards the minimizer, for general graphs and appropriate step-size. 
Second, we estimate the change of the optimal solution after a replacement, in order to evaluate its effect on the distance between the current estimate and the minimizer.
From these two elements, we derive stability conditions in open systems and establish the linear convergence of the algorithm towards a steady-state expected error. 
Our results enable to characterize the trade-off between speed of convergence and robustness to agent replacements, under the assumptions that local functions are smooth, strongly convex, and have their minimizers located in a given ball. 
The approach proposed in this paper can moreover be extended to other algorithms guaranteeing linear convergence in closed system.
\end{abstract}

\begin{IEEEkeywords}
Open multi-agent systems, distributed optimization, gradient methods, agents and autonomous systems.
\end{IEEEkeywords}

\section{Introduction}
\label{sec:intro}
Resource allocation is an important optimization problem, where a fixed amount of resources must be distributed among a specific number of activities or agents in an optimal way \cite{ibaraki1988resource,patriksson2008survey}.
In multi-agent systems, this problem is formulated as the minimization of an objective function $f$ that is separable in local costs $f_i:\R^d\to\R$ held by the agents, subject to an equality constraint on the weighted sum of the states $x_i\in\R^d$ with respect to the budget $b\in\R^d$.
The problem can then be written as
\begin{align}
    \label{eq:Statement:ResourceAllocProblem}
    &\min_{x\in\R^{nd}} f(x) = \sum_{i=1}^n f_i(x_i)&
    &\hbox{subject to}&
    &\sum_{i=1}^n a_ix_i=b,
\end{align}
where $a_i>0$ is the weight of agent $i$ to satisfying the constraint.
Each agent thus aims at minimizing its local cost while guaranteeing the budget, which requires a certain level of collaboration between them.

Problem \eqref{eq:Statement:ResourceAllocProblem} appears in different applications, including smart grids \cite{dai2021distributed}, power systems \cite{yi2016initialization}, actuator networks \cite{teixeira2013distributed}, and games \cite{liang2017distributed}.
Most of these existing studies assume that the composition of the multi-agent system remains fixed during the entire process.
Yet, with the growing size of systems nowadays, arrivals and departures of agents are expected to happen more frequently, giving rise to \emph{open multi-agent systems}, where agents join and leave the system at a time-scale similar to that of the studied process.
Consider, for instance, the case of distributed energy resources where a fixed amount of energy must be supplied by a network of devices \cite{dominguez2012decentralized}, and where some of the devices might experience failures with higher probability as the system size increases, or change their operating point due to environmental conditions.

In the framework of open systems, a fixed solution for \eqref{eq:Statement:ResourceAllocProblem} cannot be obtained as in general, the size is not fixed and the cost functions keep changing, such that the goal of the agents is to track the time-varying solution of \eqref{eq:Statement:ResourceAllocProblem} as well as possible at all times.
Moreover, as the size of such systems reaches large values, global optimization methods relying, \textit{e.g.}, on the computation of the whole gradient of $f$ are not suited since the computational complexity would be high and in some cases, it would not be practical to gather the whole gradient as agents may have entered/left in the meantime.
In fact, most of the algorithms used to solve  
\eqref{eq:Statement:ResourceAllocProblem} in a decentralized way are gradient-based, such as in \cite{xiao2006optimal}, where the authors use a weighted version of the well-known \emph{Gradient Descent} algorithm with an appropriate choice of weights to preserve the constraint.
Yet, this type of methods requires significant computational resources and, therefore, is not suitable for handling open systems. For this reason, it is important to consider optimization algorithms based on local interactions, since they are more flexible.
An alternative type of algorithms that allow to considerably reduce the computational complexity is the so-called \emph{Coordinate Descent} algorithm introduced by Nesterov, where the optimization is performed only along one direction at each iteration \cite{nesterov2012efficiency}. 
For multi-agent systems, the selection of one coordinate is equivalent to the choice of a particular edge of the network to perform the optimization. 
In such algorithms, the sequence of edges is crucial, and hence a randomized choice denoted as \emph{Random Coordinate Descent} algorithm (RCD) was studied in \cite{necoara2013random}, where convergence of the cost functions is proved under standard assumptions when only pairwise interactions are considered, so that the algorithm requires only the computation of a pair of local gradients per iteration. 

In general, interactions in a multi-agent system are characterized by an underlying network, so that agents can only communicate with a limited number of other neighbors determined by the communication network topology.
This feature of the network plays an important role in the analysis and design of algorithms, since the performance can be different in sparse and dense graphs \cite{cormen2009introduction}.

Furthermore, problems of the type of 
\eqref{eq:Statement:ResourceAllocProblem} often assume that the contributions of the agent to the constraint is homogeneous (\textit{i.e.}, that $a_i=1$ for each agent $i$).
However, this is not always the case, as for example in energy supply, where priority may be given to renewable energy sources while fossil-fuel plants ought to contribute less to the amount of energy required by each region \cite{bird2013integrating}. 

Our goal is to analyze the RCD algorithm applied to the most general possible version  
\eqref{eq:Statement:ResourceAllocProblem} (\textit{i.e.}, with arbitrary graph topologies and non-homogeneous contributions of the agents) in open systems.
In particular, we focus on systems subject to only replacements, and hence of fixed size, so that the main challenge to handle is the variations of local cost functions, such as \textit{e.g.}, in the context of energy distribution where such changes can be triggered by time-varying environmental conditions.

\subsection{Optimization in open multi-agent systems}
Algorithms for open multi-agent systems have recently been studied in several contexts. In the case of consensus, \cite{hendrickx2017open,monnoyer2022fundamental,de2020open} analyzed the behavior and performance of gossip interactions, \cite{franceschelli2020stability,dashti2019dynamic} studied dynamic consensus in terms of stability, and \cite{vizuete2020influence,varma2018open} focused on consensus with stochastic interactions.

Optimization problems in open multi-agent systems scenarios, characterized, among others, by time-varying objective functions, start getting attention as well.
In \cite{OpenDo:OpenDGDStability}, the behavior and the stability of gradient descent was studied in a setting where agents can be replaced.
In \cite{hsieh2021optimization}, an algorithm based on dual averaging was proposed to minimize a global cost function that depends on a time-varying set of active agents in a fixed size network.

Time-varying objective functions are also considered in an alternative field of work called \emph{online optimization} \cite{DO:online-varyingFunctions,shahrampour2017distributed}, where a common approach is to minimize, over a finite period of time $T$, the \emph{dynamic regret} defined as
\begin{align}\label{eq:regret_definition}
    Reg_T^d := \sum_{t=1}^T f^t(x^t)-\sum_{t=1}^Tf^t(x^{*,t}),
\end{align}
where $f^t$ is a sequence of cost functions, $x^t$ are the estimates and $x^{*,t}:=\argmin_x f^t(x)$ is the minimizer of the global function $f^t$ at time $t$.
The objective of online optimization is thus to determine the sequence of estimates $x^t$ that keeps $Reg_T^d$ as small as possible over the time period, under some assumptions about the possible sequences of time-varying cost functions.

Instead, in our problem, replacements of agents occur without any regularity and in this case, it is not possible to obtain a sublinear regret which is the usual objective of online optimization \cite{li2022survey}.
Furthermore, we can observe from \eqref{eq:regret_definition} that the computation of the regret implies an accumulation of errors from the time instant $t=1$, which does not seem appropriate for the case of open multi-agent systems, where the replacement of an agent implies that all the past information of the replaced agent is not longer available, since this agent left the system. For this reason, we perform the analysis of the resource allocation problem in open multi-agent systems considering a \emph{time-varying optimization} approach \cite{simonetto2020time}, where the objective of the algorithms is to be at all times as close as possible to the instantaneous minimizer. This objective is more suitable for open multi-agent systems where replacements may be infrequent and agents try to reach the best performance at all time instants without regarding a performance in a potential future horizon since even if the cost functions belong to the same class, there is no regularity in the way they change \cite{yang2016regularity,jadbabaie2015regularity,nali2023survey}. Nevertheless, even if replacements are not frequent, they can impact the performance of the algorithms since they modify the location of the global minimizer. 

The framework of time-varying optimization has been used in many works, including the resource allocation problem. Exact convergence can be obtained under restrictive assumptions like identical Hessian matrices \cite{wang2021distributed} or local cost functions with a time-independent part \cite{doostmohammadian20211st}, that generally guarantee the continuity of the time-varying minimizer. When exact convergence is not possible, the main challenge is the derivation of an upper bound for the error \cite{simonetto2020time}. This approach has been used, for instance, in the case of quadratic cost functions \cite{esteki2022distributed}.
However, changes of the cost functions due to replacements of agents without establishing further restrictions on the structure of the cost functions have never been explored, and this particular characteristic is the motivation of this work and plays an important role in the formulation of the problem and the derivation of an explicit bound for the error.

\subsection{Preliminary version and contribution}
We study the Random Coordinate Descent algorithm (RCD) to solve the resource allocation problem in an open system where agents get replaced during the process, relying on a decoupled analysis of the RCD algorithm and of replacements of single agents.
A preliminary version of this work was presented in \cite{monnoyer2021random}, 
where the problem was analyzed for homogeneous agents holding one-dimensional local cost functions, interacting in a fully connected network, and with uniform probabilities in the selection of edges for the updates. 
Moreover, replacements of single agents were studied by analyzing the case where possibly all agents can get replaced at once.

By contrast, in this work, we focus on heterogeneous agents holding $d$-dimensional local cost functions, interacting in a general graph topology.
We consider an arbitrary distribution for the probabilities associated to the choice of edges and we derive an upper bound for the convergence of the algorithm following a similar approach in terms of a norm induced by a matrix associated to the network.
Moreover, we now directly study the replacement of a single agent instead of considering the possibility for all agents to be replaced at once, yielding tighter bounds for that case, independent of the system size.

The remainder of this article is organized as follows. 
In Section~\ref{sec:Statement} we introduce the problem statement and the necessary preliminaries. 
In Section~\ref{sec:Replacements} we study the impact of the replacements on the location of the minimizers of the system and their distance with respect to the estimate. 
Section~\ref{sec:ClosedConv} focuses on the linear convergence of the RCD in a closed system considering an appropriate norm. 
Section~\ref{Sec:OpenSystems} presents the analysis of the RCD in an open multi-agent scenario. Finally, conclusions and future work are exposed in Section~\ref{Sec:conclusions}.

\section{Problem statement}
\label{sec:Statement}

In this section, we formulate the constrained resource allocation problem in open systems. First, we introduce the notation used along this work. Then, we present the problem in closed systems and we describe how the problem changes because of replacements, next we detail how the time is sampled. Finally, we present the Random Coordinate Descent algorithm which is considered in this work.

\subsection{Notation}
\label{sec:Notations}
For two vectors $x,y\in\R^n$, $\inProd{x}{y} = x^\top y = \sum_{i=1}^n x_iy_i$ denotes the usual Euclidean inner product and $\norm{x} = \sqrt{x^\top x}$ the Euclidean norm. The 1-norm of a vector $x$ is denoted as $\norm{x}_1$.
We denote the vector of size $n$ constituted of only ones by $\mathds{1}_n$ and the identity matrix of dimension $n$ by $I_n$. 
The vector of size $n$ constituted of only zeros is denoted by $\mathbf{0}_n$.
We use $B(x,r)=\{y:\norm{x-y}\le r \}$ to denote the closed ball of radius $r\ge 0$ centered at $x$.
For a positive (semipositive) definite matrix $A$, we denote by $\norm{x}_A=(x^\top Ax)^{1/2}$ the vector norm (seminorm) induced by $A$.
The Kronecker product is denoted by $\otimes$.

\subsection{Resource allocation problem}
\label{sec:Statement:RA}

We consider the resource allocation problem defined in \eqref{eq:Statement:ResourceAllocProblem}, where a budget $b$ must be distributed among $n$ agents according to some positive weight distribution $a\in\R^n$ (\textit{i.e.}, with $a_i>0$ for $i=1,\ldots,n$).
For the sake of simplicity, we first describe it in closed system (\textit{i.e.}, where the set of agents remains the same); we will see in the next subsection that this formulation directly extends to open systems.

The constraint in \eqref{eq:Statement:ResourceAllocProblem} can be equivalently expressed as $\prt{a^\top\otimes I_d}x=b$, where $\otimes$ denotes the Kronecker product. 
The feasible set of \eqref{eq:Statement:ResourceAllocProblem} is thus given by
\begin{equation}
    \label{eq:Statement:Sb}
    \Sab := \brc{x\in\R^{nd}|\prt{a^\top\otimes I_d}x=b}.
\end{equation}
For the particular case $d=1$, the resource allocation constraint can be expressed as $\inProd{a}{x}=b$. 
We make the following classical assumption on the local cost functions.

\begin{assumption}
\label{Ass:Statement:Fab}
    Each function $f_i$ is continuously differentiable, $\alpha$-strongly convex (\emph{i.e.}, $f_i(x)-\frac{\alpha}{2}\norm{x}^2$ is convex) 
    and $\beta$-smooth (\emph{i.e.}, $\norm{ \nabla f_i(x)- \nabla f_i(y)}\leq \beta\norm{x-y}$, $\forall x,y\in\R^d$).
\end{assumption}

Assumption~\ref{Ass:Statement:Fab} provides an upper and a lower bound to the curvature of the functions.
The value $\kappa:=\frac{\beta}{\alpha}\ge 1$ is called the \emph{condition number} of the functions.  
The set of the functions satisfying Assumption~\ref{Ass:Statement:Fab} is denoted by $\Fab$.

\begin{proposition}
If $f_1,\ldots,f_n\in\Fab$, then the global cost function $f$ from \eqref{eq:Statement:ResourceAllocProblem} satisfies $f\in\Fab$.
\end{proposition}

Since Assumption~\ref{Ass:Statement:Fab} guarantees that $f$ is $\alpha$-strongly convex, the solution of the problem \eqref{eq:Statement:ResourceAllocProblem} denoted by $x^*$ is unique \cite{boyd2004convex}.
By using KKT conditions we obtain that there exists a $\lambda^*\in\R^d$, such that a necessary and sufficient condition for the optimality of $x^*$ is
\begin{equation}
    \label{eq:Statement:Optimality}
    \nabla f(x^*) = (a^\top\otimes I_d)^\top\lambda^* 
\end{equation}
which equivalently reads $\nabla f(x^*) = a\otimes\lambda^*$,
where $\lambda^*$ is a Lagrange multiplier vector \cite{boyd2004convex}.

\subsection{Network description and open system}
\label{sec:Statement:OpenNet}
In addition to problem \eqref{eq:Statement:ResourceAllocProblem}, we assume to have an undirected and connected network $\Gcal = (\Vcal,\Ecal)$ where the set of nodes is given by $\Vcal = \brc{1,\ldots,n}$ and the set of edges by $\Ecal\subseteq \Vcal\times\Vcal$.
Each agent $i\in\Vcal$ has access to a local cost function $f_i:\R^d\to\R$ and to a local variable $x_i\in\R^d$.
Agents can exchange information at random times through pairwise interactions according to the network $\Gcal$. 
Whenever an interaction happens in the system, an edge $(i,j)\in\Ecal$ is selected with some fixed probability $p_{ij}>0$ and agents $i$ and $j$ can then exchange information in a bidirectional manner to update their respective estimates.

Moreover, we consider that \emph{replacements of agents} happen in the system, making it \emph{open}.
Each agent $i\in\Vcal$ gets replaced at random time instants, resulting in the change of its local cost function, and hence of the global minimizer $x^*$.
Following the approach in \cite{OpenDo:OpenDGDStability}, we restrict the location of the minimizers of the local functions:

\begin{assumption}
\label{Ass:Statement:B(0,c)&f*=0}
    There exists $c>0$, such that for all $i\in\Vcal$, the minimizer of $f_i$ denoted as $\bar x_i^*:=\argmin_{x}f_i(x)$ satisfies $\bar x_i^*\in B(\mathbf 0_d,c)$.
    Moreover, without loss of generality $f_i(\bar x_i^*)=0$ for all $i\in\Vcal$.
\end{assumption}

Assumption~\ref{Ass:Statement:B(0,c)&f*=0} guarantees a certain level of uniformity among the local cost functions.
In particular, it prevents arbitrary changes of functions, and thus of $x^*$, during replacements. Also, since our objective is to study the convergence of the minimizer and how it is affected by the replacements, we do not use the actual values of $f_i$.

For the sake of simplicity, we assume that when agent $i$ is replaced, the joining agent that takes its place retrieves its label $i$ and its estimate $x_i$, so that the constraint $\prt{a^\top\otimes I_d}x=b$ is preserved, but receives a new local cost function satisfying Assumptions~\ref{Ass:Statement:Fab} and \ref{Ass:Statement:B(0,c)&f*=0}.
Denoting $f_i^k$ the local cost function held by agent $i$ at the time instant $k$, we can then reformulate \eqref{eq:Statement:ResourceAllocProblem} as the time-varying resource allocation problem
\begin{equation}
    \label{eq:Statement:Objective:OpenResourceAllocProblem}
    \min_{x\in {\Sab}} f^k(x) := \sum_{i=1}^n f_i^k(x_i),
\end{equation}
where the changes of the cost functions are due to replacements.
The solution of \eqref{eq:Statement:Objective:OpenResourceAllocProblem} can thus differ from a time instant $t_k$ to another, and we denote it by $\optx{k} := \argmin_{x\in \Sab} f^k(x)$.
The objective of the agents is to track $\optx{k}$ as well as possible even though replacements happen in the system. 

\subsection{Discrete-event modelling}
\label{sec:Statement:DT&Goal}
The evolution of the open network presented in the previous section is characterized by the instantaneous occurrence at random time instants of either pairwise interactions or replacements.
With a slight abuse of language, we call \quotes{event} such occurrence which results in a modification of the system depending on its nature.
We moreover define the \quotes{\emph{event set}} of the system from which that nature is drawn as
\begin{equation}
    \label{eq:Statement:EventSet}
    \Xi = U \cup R,
\end{equation}
where $U := \bigcup_{(i,j)\in\Ecal}U_{ij}$ is the set of all possible events $U_{ij}$, \textit{i.e.}, the pairwise interaction between two connected agents $i$ and $j$, and $R := \bigcup_{i\in\Vcal}R_i$ is the set of all events $R_i$, \textit{i.e.}, the replacement of a single agent $i$ in the system.
This formulation actually relates with \emph{discrete event systems} (see \emph{e.g.}, \cite{kurve2013discreteEvent}).

We assume that two events never occur simultaneously, so that we can consider a discrete evolution of the time where each time-step $k\in\N$ corresponds to the time instant at which the $k$-th event takes place.
For all $k\in\N$, we then define the random variable $\xi_k\in\Xi$ which characterizes the nature of the event taking place at the time-step $k$.
We moreover consider the following assumption that guarantees that replacements and interactions are independent processes, so that the event happening at time-step $k$ is a pairwise interaction (\textit{i.e.}, $\xi_k\in U$) with fixed probability $p_U$, and a replacement (\textit{i.e.}, $\xi_k\in R$) with fixed probability $p_R = 1-p_U$.
\begin{assumption}
\label{Ass:Statement:Indep}
    For all $k\in\N$, $\xi_k$ is independent of $\xi_j$ for all $j\neq k$, and of any variable in the system prior to time-step $k$, such as the estimates or local cost functions.
\end{assumption}

Our goal is to characterize the evolution of the distance between the estimates held by the agents $x^k$ and the instantaneous minimizer $\optx{k}$.
A choice of measure for this distance is, \textit{e.g.}, the squared Euclidean norm $\norm{x^k-\optx{k}}^2$, although we will see later that this metric might not be appropriate for general graph topologies.
Assumption~\ref{Ass:Statement:Indep} allows for performing this analysis by studying independently the effect of each type of event on our metric in order to characterize its evolution in expectation over a single time step.

\subsection{Random Coordinate Descent (RCD) algorithm}
\label{sec:Statement:RCD}

To compute the solution of \eqref{eq:Statement:Objective:OpenResourceAllocProblem} we consider the Random Coordinate Descent (RCD) algorithm introduced in \cite{necoara2013random}.
This algorithm involves the update of the states of only a pair of neighbouring agents at each iteration, so that it is distributed and its computational complexity is cheap.
Hence, in the event $U_{ij}$, \textit{i.e.}, whenever the pair of agents $(i,j)\in\Ecal$ is selected with probability $p_{ij}$ during a pairwise update event,
they perform an \emph{RCD update}, which is defined as follows for some nonnegative step-size $h\geq0$:
\begin{equation}
\label{eq:Statement:RCD:UpdateRule}
    x^+ = x - h\Qbf^{ij}\nabla f(x),
\end{equation}
where $\Qbf^{ij}$ is the $nd \times nd$ matrix defined as $\Qbf^{ij}=Q^{ij}\otimes I_d$, with $Q^{ij}$ the $n \times n$ matrix filled with zeroes except for the four following entries:
\begin{align*}
    &[Q^{ij}]_{i,i} = \frac{a_j^2}{a_i^2+a_j^2};&&
    [Q^{ij}]_{i,j} = -\frac{a_ia_j}{a_i^2+a_j^2};\\
    &[Q^{ij}]_{j,i} = -\frac{a_ia_j}{a_i^2+a_j^2};&
    &
    [Q^{ij}]_{j,j} = \frac{a_i^2}{a_i^2+a_j^2}
    .
\end{align*}
With the update rule \eqref{eq:Statement:RCD:UpdateRule}, only agents $i$ and $j$ update their estimates while all the other agents keep it the same.
For agents $i$ and $j$, \eqref{eq:Statement:RCD:UpdateRule} essentially amounts to perform a gradient step on the function $f_i(x_i) + f_j(x_j)$ under the constraint that $a_ix_i+a_jx_j$ remains constant.
This ensures that the resource allocation constraint is preserved as long as the starting point satisfies it. 
In particular, in the case of homogeneous agents (\textit{i.e.}, where $a = \mathds1_n$), then one shows that 
$x_i^+ = x_i - \frac{h}{2}(\nabla f_i(x_i)-\nabla f_j(x_j))$, so that the update follows both gradients with equal weight while preserving the constraint.

Observe that the method presented here requires 
(i) the initial point to be feasible (which is rather standard for such methods), 
and (ii) that the estimates are maintained during replacements (which is assumed in Section~\ref{sec:Statement:OpenNet}).
Otherwise one would need to design a process to run in parallel of the optimization process to meet the constraint.
This is, however, out of the scope of this paper.

\begin{remark}
\label{rem:Statement:RCD:Derivation}
The update rule \eqref{eq:Statement:RCD:UpdateRule} can be formally obtained by solving the following optimization problem, which corresponds to the interpretation given above (we refer to \cite{necoara2013random} for details):
\begin{equation}
\label{eq:Statement:RCD:minProb}
    \arg\min_{s_i,s_j\in\R^d:a_is_i+a_js_j=0}
    \left\langle
    \begin{bmatrix}
        \nabla f_i(x_i)\\ \nabla f_j(x_j)
    \end{bmatrix},
    \begin{bmatrix}
        s_i\\s_j
    \end{bmatrix}\right\rangle
    + \frac\beta2\left\lVert
    \begin{bmatrix}
        s_i\\s_j
    \end{bmatrix}
    \right\lVert^2.
\end{equation}
Based on the approach of \cite{necoara2013random}, one can then show that the optimal step-size that solves \eqref{eq:Statement:RCD:minProb} is given by $h = 1/\beta$.
\end{remark}

We also introduce the following matrix that builds on the definition of the update rule \eqref{eq:Statement:RCD:UpdateRule} and that will be used later:
\begin{equation}
    \label{eq:Statement:RCD:Lpbf}
    \Lpbf 
    = \sum_{(i,j)\in\Ecal} p_{ij}\Qbf^{ij} 
    = \prt{\sum_{(i,j)\in\Ecal}p_{ij}Q^{ij}}\otimes I_d 
    = L_p\otimes I_d.
\end{equation}
This matrix appears in the dynamics corresponding to the conditional expectation:
\begin{equation}\label{eq:expected_RCD_algorithm}
\Ep{x(k+1)|x(k)}=x(k)-h\Lpbf \nabla f(x(k)),    
\end{equation}
and will be used for the definition of an appropriate norm for the analysis of the RCD algorithm. 
Observe that by definition of $Q^{ij}$ and $L_p$, we have
\begin{equation}
    \label{eq:statement:RCD:Qa=La=0}
    L_pa = Q^{ij}a = \mathbf{0}_n,
\end{equation}
which means that zero is an eigenvalue of both $Q^{ij}$ and $L_p$ with corresponding eigenvector $a$. 
We denote by $\lambda_2$ and $\lambda_n$ respectively the second smallest and the largest eigenvalues of $L_p$.
Since $L_p$ is symmetric, all the eigenvalues are real and satisfy $0=\lambda_1<\lambda_2\le\cdots\le\lambda_n$ when the graph $\Gcal$ is connected (we refer to Lemma 3.3 of \cite{necoara2013random} for a detailed proof).

\begin{remark}
For a graph $\Gcal=(\Vcal,\Ecal)$, when $a=\mathds1_n$ (homogeneous agents) and the probabilities $p_{ij}$ are uniformly distributed, we have $L_p = \frac{1}{2|\Ecal|}L$, where $L$ is the usual Laplacian of the graph.
Hence, we refer to $L_p$ as a \quotes{scaled Laplacian}, as it enjoys similar properties, especially in terms of eigenvalues.
\end{remark}

\section{Effect of replacements} 
\label{sec:Replacements}
In this section, we bound the distance by which the minimizer of $f$ can change after the replacement of a single agent, \textit{i.e.}, the modification of a single function. 
Our first two results concern the location of the minimizer: 
Lemma~\ref{lem:Repl:minLoc:Global} is a generalisation of the analysis performed in \cite{monnoyer2021random}, and Lemma~\ref{lem:Repl:minLoc:Local} studies the location of the minimizer held by each individual agent.

\begin{lemma}
\label{lem:Repl:minLoc:Global}
    Let $x^*:=\argmin_{x\in\Sab}\sum_{i=1}^nf_i(x_i)$. 
    If all $f_i$ satisfy Assumptions~\ref{Ass:Statement:Fab} and \ref{Ass:Statement:B(0,c)&f*=0}, then $x^*\in B(\mathbf 0_{nd},R_{b,\kappa})$ with
    \begin{equation}
        \label{eq:lem:Repl:minLoc:Global}
        R_{b,\kappa} = \sqrt{n\kappa} \prt{c+\frac{c}{\sqrt\kappa}+\frac{\norm{b}}{\sqrt{n}\norm{a}}},
    \end{equation}
    where $c$ was defined in Assumption~\ref{Ass:Statement:B(0,c)&f*=0}.
\end{lemma}
\begin{proof}
    The proof is left to Appendix~\ref{appendix:minLocLemmas1}.
\end{proof}

\begin{lemma}
\label{lem:Repl:minLoc:Local}
    Let $x^* := \argmin_{x\in\Sab}  \sum_{i=1}^n f_i(x_i)$. If $f_i$ satisfies Assumptions~\ref{Ass:Statement:Fab} and \ref{Ass:Statement:B(0,c)&f*=0} for all $i=1,\ldots,n$, then 
    for $\lambda^*$ defined in \eqref{eq:Statement:Optimality}
    \begin{equation}\label{eq:equation2_proof_bound_minimizer}
        \norm{\lambda^*}\le\frac{\beta}{\norm{a}^2}\prt{\norm{b}+c\norm{a}_1};
    \end{equation}
    and
    \begin{equation}\label{eq:lem:Repl:minLoc:Local:nolambda}
        \norm{x_i^*} \le \frac{a_i}{\norm{a}^2}\kappa\prt{\norm{b}+c\norm{a}_1}+c.
    \end{equation}
\end{lemma}
\begin{proof}
    The proof is left to Appendix~\ref{appendix:minLocLemmas2}.
\end{proof}

We can now use these two lemmas to characterize the evolution of the distance between the estimate $x^k$ and the minimizer $\optx{k}$ after a replacement event. 
Without loss of generality, we assume that agent $n$, and hence $f_n$, is replaced, and for the $n+1$ functions $f_1,f_2,\ldots,f_{n-1},f_n^{(1)},f_n^{(2)}$ satisfying Assumptions~\ref{Ass:Statement:Fab} and \ref{Ass:Statement:B(0,c)&f*=0} we define the minimizer before a replacement $x^{(1)}$, and after a replacement $x^{(2)}$ as
\begin{align}
    \label{eq:Repl:x1&x2}
    x^{(1)} 
    &:= \argmin_{x\in \Sab} \prt{\sum\nolimits_{i=1}^{n-1}f_i(x_i) + f_n^{(1)}(x_n)};\nonumber\\
    x^{(2)} 
    &:= \argmin_{x\in \Sab} \prt{\sum\nolimits_{i=1}^{n-1}f_i(x_i) + f_n^{(2)}(x_n)}.
\end{align}

\begin{proposition}
\label{prop:Repl:minChange}
    Consider $x^{(1)}$ and $x^{(2)}$ as defined in \eqref{eq:Repl:x1&x2}, let $\amax$ and $\amin$ respectively denote the largest and smallest values in $a$, and let $\rho_a := \frac{\amax^2}{\norm{a}^2-\amax^2}$, then 
    \begin{equation}
    \label{eq:Bound_Mnk}      
      \norm{x^{(1)}-x^{(2)}}^2
      \le \min\{\psi_{n,\kappa},\chi_{n,\kappa},\theta_{n,\kappa}\}
      =:\bar M_{n,\kappa}^2,
    \end{equation}
    with
    \vspace{-0.3cm}
    {\small
    \begin{align}
        \label{eq:prop:Repl:minChange:psi}      
        \psi_{n,\kappa} 
        &= 4n\kappa \prt{c+\frac{c}{\sqrt\kappa}+\frac{\norm{b}}{\sqrt{n}\norm{a}}}^2;\\
        \label{eq:prop:Repl:minChange:chi}
        \chi_{n,\kappa} 
        &= 8\prt{\frac{\amax}{\norm{a}^2}\kappa(\norm{b}+c\norm{a}_1)+c}^2;\\
        \label{eq:prop:Repl:minChange:theta}
        \theta_{n,\kappa} 
        &= 4\prt{1+\frac{(\kappa+1)^2}{4\kappa}\rho_a} \prt{\frac{\amax}{\norm{a}^2}\kappa(\norm{b}+c\norm{a}_1)+c}^2.
    \end{align}}
\end{proposition}
\begin{proof}
    The proof is left to Appendix~\ref{sec:Appendix:MinChange:Proof}.
\end{proof}

The bound $\bar M_{n,\kappa}^2$ from Proposition~\ref{prop:Repl:minChange} is obtained by taking the minimum between three quantities: $\psi_{n,\kappa}$, $\chi_{n,\kappa}$ and $\theta_{n,\kappa}$.
The first one follows from the largest possible distance existing between two minimizers, defined by the region in which they can be located.
The second and third ones rely on the largest possible distance between the local minimizers corresponding to the replaced agents. While $\chi_{n,\kappa}$ and $\theta_{n,\kappa}$ are derived using inequalities associated with $\alpha$-strongly convex functions, the proof of $\theta_{n,\kappa}$ also involves the use of additional properties corresponding to $\beta$-smooth functions and the determination of the maximum value of a concave function. The bound $\theta_{n,\kappa}$ shows a strong dependence on the weights of the agents through the coefficient $\rho_a$, which is not present in the other two bounds. Notice that the bounds $\chi_{n,\kappa}$ and $\theta_{n,\kappa}$ coincide when
$$
\frac{\prt{\kappa+1}^2}{4\kappa}\rho_a=1.
$$

Let $\bar a$ and $\overline{a^2}$ respectively stand for the average value and average of the squared values of $a$.
One can more generally highlight the dependencies of the three quantities with the parameters using standard algebraic manipulations, yielding
{\small
\begin{align*}
    \psi_{n,\kappa}
    &\leq 4n\kappa\prt{2c+\frac{\norm{b}}{n\bar a}}^2
    = O(n\kappa);\\
    \chi_{n,\kappa}
    &\leq 8\prt{\frac{\amax^2}{\overline{a^2}}\prt{\frac{\norm{b}}{n\bar a}+2c}\kappa}^2
    = O(\kappa^2);\\
    \theta_{n,\kappa}
    &\leq 4\prt{\tfrac{\amax^2}{\amin^2}\prt{\tfrac{\kappa}{2(n-1)}+2}}\prt{\tfrac{\amax^2}{\overline{a^2}}\prt{\tfrac{\norm{b}}{n\bar a}+2c}\kappa}^2
    = O\prt{\kappa^2+\frac{\kappa^3}{n}}.
\end{align*}}
The linear scaling of $\psi_{n,\kappa}$ in both $n$ and $\kappa$ and the higher order scaling of both $\chi_{n,\kappa}$ and $\theta_{n,\kappa}$ in only $\kappa$ suggest that $\psi_{n,\kappa}$ is tighter for small values of $n$ and large values of $\kappa$, whereas $\theta_{n,\kappa}$ and $\chi_{n,\kappa}$ are tighter otherwise.
The main difference between $\chi_{n,\kappa}$ and $\theta_{n,\kappa}$ lies in a multiplicative factor, constant for the former, and depending of the parameters and the values in $a$ for the latter.
In general, $\chi_{n,\kappa}$ tends to be tighter than $\theta_{n,\kappa}$ as $\kappa$ gets large and $n$ small.
This difference becomes significant in heterogeneous settings, where it can get tighter than $\psi_{n,\kappa}$ as well.
These behaviors are illustrated in Fig.~\ref{fig:Repl:minChange:Mnk:HeterogVsHomog}.

\begin{figure}[h!]
    \centering
    \includegraphics[width=0.45\textwidth,clip = true, trim=1.5cm 10cm 1.8cm 10.25cm,keepaspectratio]{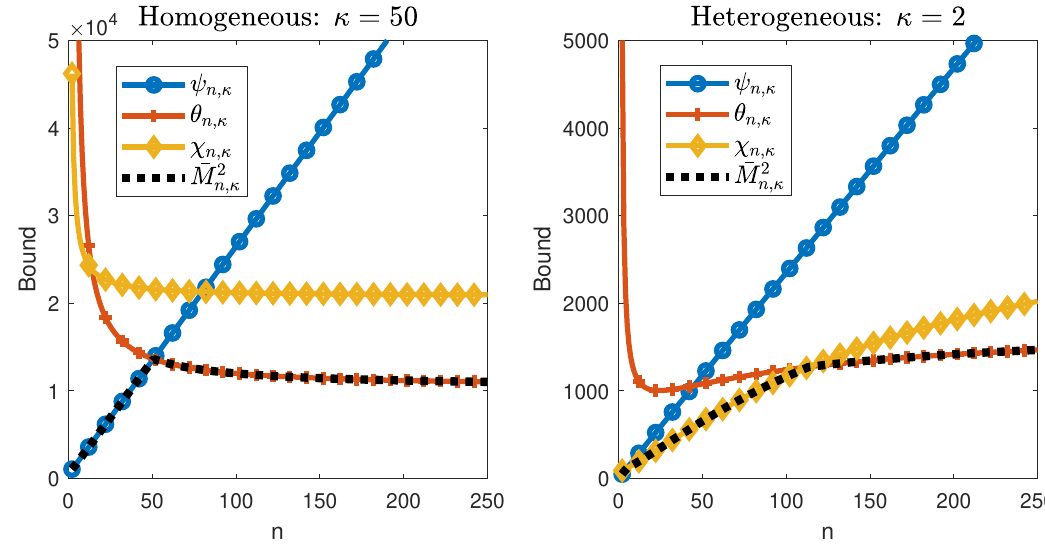}
    \caption{Bounds of Proposition~\ref{prop:Repl:minChange} with respect to the system size $n$ for $b=1$, $c=1$, respectively for $\kappa=50$ with homogeneous agents ($a_i=1$ for all $i$) on the left, and $\kappa=2$ with heterogeneous agents ($a_1=10$, $a_i=1$ for $i>1$) on the right.
    The plots show all three quantities $\psi_{n,\kappa}$, $\theta_{n,\kappa}$ and $\chi_{n,\kappa}$ as well as the final bound $\bar M_{n,\kappa}^2$ for both cases.}
    \label{fig:Repl:minChange:Mnk:HeterogVsHomog}
\end{figure}

\begin{remark}
    \label{rem:minChange:scalings}
    The interpretation of the quantities $\psi_{n,\kappa}$, $\chi_{n,\kappa}$ and $\theta_{n,\kappa}$ actually depends on the implicit assumption that $\norm{b}$ is fixed and $\norm{a}_1$ scales with $n$ (\textit{i.e.}, $\bar a$ is fixed).
    This particular modelling choice is arbitrary, and implies that the solution held by an agent $\optx{k}_i$ becomes smaller for large values of $n$.
    Other choices might have different implications on the interpretation, and in particular on the scaling of these quantities.
    For instance one could choose to either fix $\norm{b}$ and $\norm{a}_1$, or that both $\norm{b}$ and $\norm{a}_1$ scale with $n$, so that the $\optx{k}_i$ remain mostly the same no matter $n$ (observe that the latter choice yields the same scalings than those presented above).
\end{remark}

The result of Proposition~\ref{prop:Repl:minChange} can be analyzed with respect to empirical results derived with the PESTO toolbox \cite{PESTO}, which allows computing exact empirical bounds for quantities related to convex functions. 
A similar analysis was performed in \cite{monnoyer2021random}, and we thus refer to Appendix B of \cite{monnoyer2021random} for details about the PESTO simulation.
For the sake of simplicity, the analysis here is only done for the homogeneous case, and consequently does not involve $\chi_{n,\kappa}$; similar conclusions could however be drawn the same way using heterogeneous agents.

We can observe in Fig.~\ref{fig:PESTOComp_wrtKappa} that even though there is some gap between the theoretical result and that obtained using PESTO, the scaling of the bounds with respect to $n$ and $\kappa$ is well captured.
In particular, the top plot shows that $\bar M^2_{n,\kappa} = \theta_{n,\kappa}$ when $n$ becomes large, resulting in the convergence of $\bar M^2_{n,\kappa}$ towards a constant, consistently with the result obtained with PESTO.
In parallel, the bottom plot suggests that the bounds from PESTO asymptotically grows linearly with $\kappa$, consistently with the evolution of $\psi_{n,\kappa}$, which is the value taken by $\bar M_{n,\kappa}^2$ for large values of $\kappa$.

\begin{figure}
    \centering
    \includegraphics[width=0.5\textwidth,clip = true, trim=1cm 10cm 1cm 10cm,keepaspectratio]{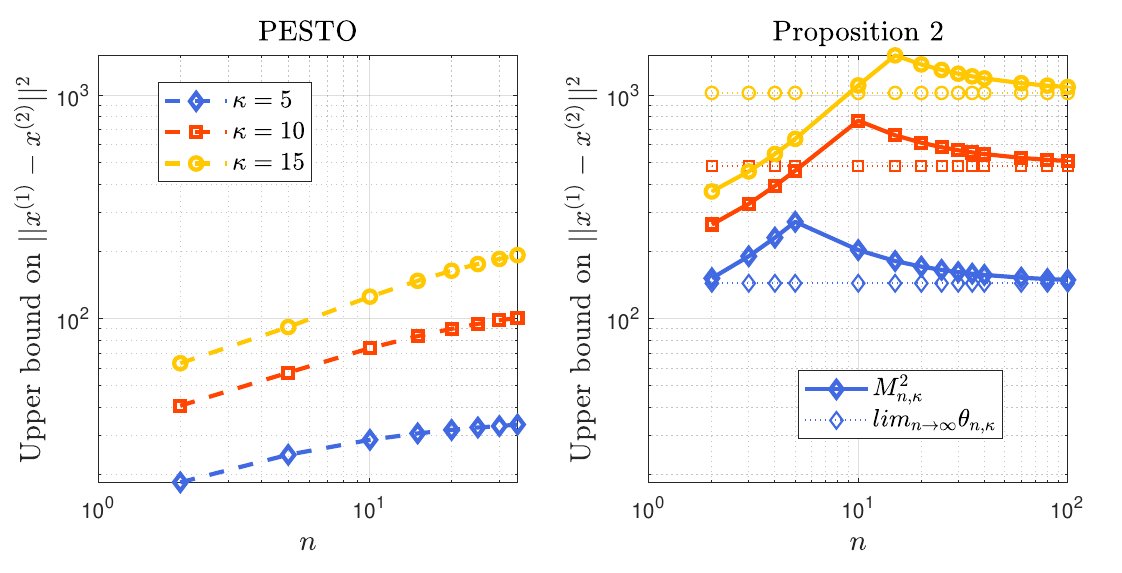}
    \includegraphics[width=0.5\textwidth,clip = true, trim=1cm 10cm 1cm 10.5cm,keepaspectratio]{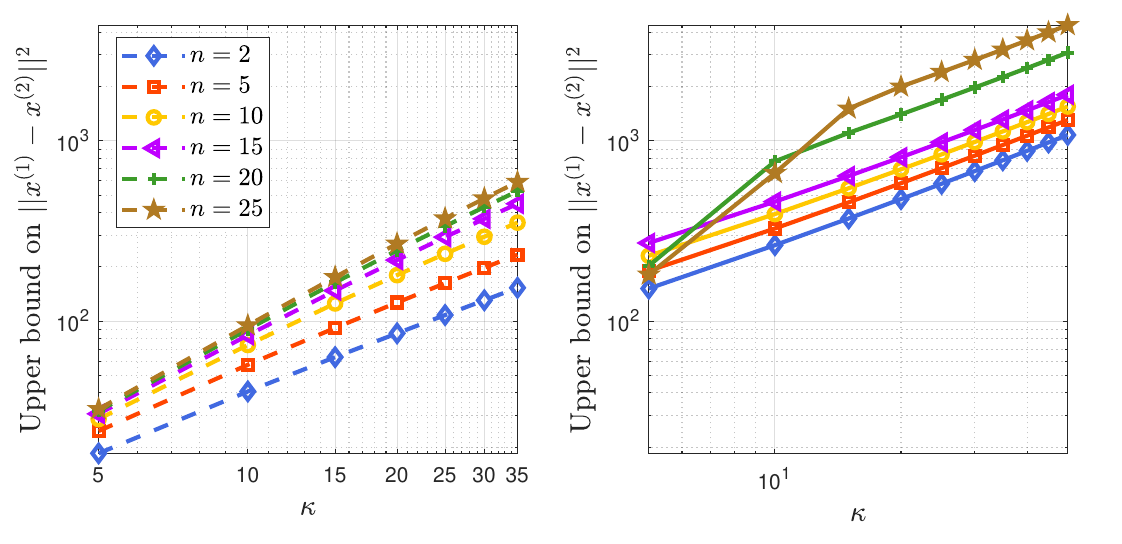}
    \caption{The upper bound \eqref{eq:Bound_Mnk} on $\norm{x^{(1)}-x^{(2)}}^2$ respectively with respect to $n$ with homogeneous agents for several values of $\kappa$ (top) and with respect to $\kappa$ for several values of $n$ (bottom).
    For each plot the bound obtained in Proposition~\ref{prop:Repl:minChange} (right) is compared with the empirical upper bound derived using PESTO in the same settings (left). 
    The top-right plot also shows the asymptotic value expected to be reached by $\theta_{n,\kappa}$ as $n\to\infty$ based on \eqref{eq:prop:Repl:minChange:theta}.}
    \label{fig:PESTOComp_wrtKappa}
\end{figure}

\section{Linear convergence of RCD in closed system}
\label{sec:ClosedConv}

We now analyze the effect of the second type of events happening in the system, \textit{i.e.}, pairwise interactions resulting in RCD updates.
This corresponds to studying the linear convergence of the RCD Algorithm in closed system.

\subsection{Linear convergence and $\Lpbf^\dag$-seminorm}
\label{sec:ClosedConv:LdagNorm}
In this section, we derive the constant of convergence of the RCD algorithm in terms of the distance to the minimizer with the objective of characterizing the effect of a single RCD step on that expected distance at interaction events. 
We introduce the following standard definitions \cite{ortega2000iterative}.

\begin{definition}[Q-Linear Convergence]
\label{def:ClosedConv:LinConv}
    Let $\{x^k\}$ be the sequence of points converging to some point $x^* \in \R^d$ generated by some algorithm. For any norm $\norm{\cdot}$,
    we say the convergence is \emph{Q-linear} if there exists $r\in(0,1)$ such that for all $k$
    $$
    \norm{x^{k+1}-x^*}\le r \norm{x^k-x^*}.
    $$
The number $r$ is called the \textit{constant of convergence}.
\end{definition}

\begin{definition}[R-Linear Convergence]
    Let $\{x^k\}$ be the sequence of points converging to some point $x^* \in \R^d$ generated by some algorithm. For any norm $\norm{\cdot}$,
    we say the convergence is \emph{R-linear}  if there exists $r\in(0,1)$ and some positive constant $C$ such that \linebreak for all $k$
    $$
    \norm{x^{k}-x^*}\le Cr^k.
    $$
\end{definition}
In the rest of the work, we will refer to Q-linear convergence just as linear convergence. R-linear convergence is typically referred as exponential convergence in control systems theory. Clearly, R-linear convergence is weaker than linear convergence since it is concerned with the overall rate of decrease in the error, rather that the decrease over each individual iteration of the algorithm \cite{nocedal2006numerical}.

In \cite{necoara2013random}, the author proves linear convergence of the RCD algorithm in expectation in terms of the function value, \textit{i.e.}, $f(x)-f(x^*)$. 
Hence, from the inequalities corresponding to smooth functions and strong convexity \cite{nesterov2018lectures,MAL-050}, it is straightforward to prove R-linear convergence of the algorithm from \cite[Eq. (26)]{necoara2013random}:
\begin{equation}
\label{eq:ClosedConv:LdagNorm:ExpConv}
    \E\brk{\norm{x^k-x^*}}\le \kappa (1-\alpha\lambda_2)^k\norm{x_0-x^*}.
\end{equation}
However, due to the alternation of updates and replacements, our analysis in open systems requires the strict contraction of some metric after each iteration. 
The linear convergence of the RCD algorithm was established in the preliminary work \cite{monnoyer2021random} for the Euclidean norm under the assumption of a complete communication graph with homogeneous agents and uniform probabilities $p_{ij}$. 
Nevertheless, the following example shows that such contraction no longer holds for the Euclidean norm for general graphs.

\begin{example}
\label{example:noLinConv}
    Consider a line graph with 3 agents satisfying the constraint $\inProd{\mathds{1}}{x}=-3$ with probabilities $p_{12}=0.9$, $p_{23}=0.1$ (and hence $p_{13}=0$), and whose local cost functions and estimates at iteration $k$ are: 
    \begin{center}
        \begin{tabular}{ |c|c|c|c| } 
            \hline
            $i$ & $f_i(x_i)$ & $x_i^*$ & $x_i^k$ \\ 
            \hline
            $1$ & $50(x_1-2)^2$ & $2$ & $10$ \\
            \hline
            $2$ & $20(x_2+2)^2$ & $-2$ & $7$ \\
            \hline
            $3$ & $(x_3+3)^2$ & $-3$ & $-20$ \\ 
            \hline
        \end{tabular}
    \end{center}
    Starting from $x^k$ the expected result of the RCD operation with step-size $h=1/\beta=0.01$ is
    \begin{equation}
        \label{eq:example:noLinConv}
        \Ep{\norm{x^{k+1}-x^*}^2} = 437.204 > 434 = \norm{x^k-x^*}^2,
    \end{equation}
    and hence linear convergence cannot be achieved.
\end{example}

\begin{remark}[Weighted gradient descent]
Notice that the expected behavior of the RCD algorithm \eqref{eq:expected_RCD_algorithm} is linked with the weighted gradient descent, whose convergence has been studied \textit{e.g.}, in \cite{xiao2006optimal,lakshmanan2008decentralized,cherukuri2015distributed}.
Those works, however, do not prove \emph{linear} convergence, which is required for the approach we follow in this paper.
\end{remark}

For this reason, we propose to study the problem in a different norm associated with the algorithm. 
Since the RCD is performed along a network of agents, a natural choice is to consider norms induced by associated matrices as in \cite{wang2021laplacian}.
In this case, we focus on the seminorm induced by the Moore-Penrose inverse of the matrix $\Lpbf$ introduced in \eqref{eq:Statement:RCD:Lpbf}, denoted by $\Lpbf^{\dagger}$, and defined as follows for some $x\in\R^{nd}$:
\begin{equation}
    \label{eq:ClosedConv:LdagNorm:LdagNormDef}
    \LdagNorm{x} := \sqrt{x^\top\Lpbf^{\dag}x}.
\end{equation}
We show with the next proposition that this seminorm is a norm on $\Sao$, where we recall that $\Sao$ is the feasible set defined in \eqref{eq:Statement:Sb} when $b=\mathbf{0}_d$ and corresponds to the kernel of $a^T\otimes I_d$. For the particular case $d=1$, $\Sao$ is the orthogonal complement of $a$.

\begin{proposition}
\label{prop:seminorm_norm}
    The seminorm $\LdagNorm{\cdot}$ is a norm on $\Sao$.
\end{proposition}
\begin{proof}
    By \eqref{eq:ClosedConv:LdagNorm:LdagNormDef}, $\LdagNorm{x}=0$ implies that $x$ must be in the kernel of the matrix $\Lpbf^{\dag}$. From \eqref{eq:Statement:RCD:Lpbf} we have $\Lpbf^{\dag}=L_p^{\dag}\otimes I_d$, and since the kernel of $\Lpbf^{\dag}$ is spanned by all the eigenvectors of $\Lpbf^{\dag}$ corresponding to a zero eigenvalue, we get
    $$
    \mathrm{ker}(\Lpbf^{\dag})=\{ x\in\R^{nd}\, |\, x= a\otimes w, w\in\R^d \}.
    $$ 
    Since $x\in\Sao$, it must satisfy $(a^\top\otimes I_d)x=\mathbf{0}_{nd}$ and we have:
    $$
    \prt{a^\top\otimes I_d}x=\prt{a^\top\otimes I_d}\prt{a\otimes w}=\norm{a}^2w,
    $$
    which is equal to $\mathbf{0}_d$ only for $w=\mathbf{0}_d$.
\end{proof}

If $x,y \in \Sab$, then $z=x-y$ belongs to $\Sao$, so that the norm $\LdagNorm{\cdot}$ can be used to measure the distance between two vectors in the context of this work.

\subsection{Contraction of an iteration in closed system}
\label{sec:ClosedConv:General}
Let us remind the update rule of the RCD algorithm defined in \eqref{eq:Statement:RCD:UpdateRule} for some positive step-size $h$ as
\begin{equation}
    \label{eq:ClosedConv:rule}
    x^+ = x - hQ^{ij}\nabla f(x).
\end{equation}
In the following proposition, we analyze the convergence of \eqref{eq:ClosedConv:rule} with respect to the norm induced by $\Lpbf^\dag$ defined in the previous section.

\begin{proposition}
\label{prop:ClosedConv:kLRate}
    Let a function $f(x):= \sum_{i=1}^n f_i(x_i)$ and $x^*:= \argmin_{x\in \Sab} f(x)$.
    Under Assumption~\ref{Ass:Statement:Fab}, for any positive scalar 
    \begin{equation}
    \label{eq:prop:ClosedConv:kLRate:h}
        h\leq \frac{\lambda_2}{\lambda_n}\frac{2}{\alpha+\beta},    
    \end{equation}
    and for any initial point $x\in \Sab$, then the update rule \eqref{eq:ClosedConv:rule}
    applied on the randomly selected pair of agents $(i,j)\in \Ecal$ satisfies
    \vspace{-0.1cm}
    \begin{equation}
        \label{eq:prop:ClosedConv:kLRate}
        \E\brk{\LdagNorm{x^+-x^*}^2} 
        \leq \prt{1-2h\alpha\lambda_2+h^2\alpha^2\lambda_n}\LdagNorm{x-x^*}^2.
    \end{equation}
\end{proposition}
\begin{proof}
    By definition:
    \begin{align}
        \!\!\!\!\E\brk{\LdagNorm{x^+-x^*}^2} 
        &= \sum_{(i,j)\in\Ecal}p_{ij}\LdagNorm{x-h\Qbf^{ij}\nabla f(x)-x^*}^2\nonumber\\
        &= \LdagNorm{x-x^*}^2 + h^2\sum_{(i,j)\in\Ecal}p_{ij}\LdagNorm{\Qbf^{ij}\nabla f(x)}^2\nonumber\\
        \label{eq:proof:prop:ClosedConv:kLRate:Start}
        &\quad - 2h\!\sum_{(i,j)\in\Ecal}p_{ij}\inProd{\Qbf^{ij}\nabla f(x)}{L_p^\dag(x-x^*)}.
    \end{align}
    
    \noindent Since $\Lpbf = \sum_{(i,j)\in\Ecal}p_{ij}\Qbf^{ij}$, it follows that
    \begin{align}
        \label{eq:proof:prop:ClosedConv:ExpandECk1}
        \E\brk{\LdagNorm{x^+-x^*}^2} 
        &= \LdagNorm{x-x^*}^2 + h^2\sum_{(i,j)\in\Ecal}p_{ij}\LdagNorm{\Qbf^{ij}\nabla f(x)}^2\nonumber\\
        &\ \ \ - 2h\inProd{\Lpbf\nabla f(x)}{\Lpbf^\dag(x-x^*)}.
    \end{align}
    
    \noindent We first treat the second term of the right-hand side of \eqref{eq:proof:prop:ClosedConv:ExpandECk1}.
    Remember from \eqref{eq:statement:RCD:Qa=La=0} 
    that $Q^{ij}a=\mathbf{0}_n$, and from \eqref{eq:Statement:Optimality} that $\nabla f(x^*) = a\otimes\lambda^*$ for some $\lambda^*\in\R^d$.
    Hence, since $\Qbf^{ij}=Q^{ij}\otimes I_d$ by definition:
    \begin{align}
        \label{eq:proof:prop:ClosedConv:kLRate:Qijnablaf(x^*)=0}
        \Qbf^{ij}\nabla f(x^*) 
        &= (Q^{ij}\otimes I_d)(a\otimes \lambda^*)
        = (Q^{ij}a)\otimes \lambda^* 
        = \mathbf{0}_{nd}.
    \end{align}
    It thus follows that
    \begin{align}
        \LdagNorm{\Qbf^{ij}\nabla f(x)}^2
        &= \LdagNorm{\Qbf^{ij}(\nabla f(x)-\nabla f(x^*))}^2\nonumber\\
        \label{eq:proof:prop:ClosedConv:kLRate:UBquadTermInter}
        &\leq \frac{1}{\lambda_2}\norm{\Qbf^{ij}(\nabla f(x)-\nabla f(x^*))}^2,
    \end{align}
    where the inequality follows from the fact that the eigenvalues of $\Lpbf$ are exactly those of $L_p$ repeated $d$ times (by Theorem~13.12 of \cite{KronEig}), so that the smallest and largest nonzero eigenvalues of $\Lpbf^\dag$ are respectively $1/\lambda_n$ and $1/\lambda_2$, yielding for all $z\in\R^{nd}$:
    \begin{equation}
        \label{eq:proof:prop:ClosedConv:kLRate:LpdagNorm(z)<=Norm(z)/l2}
        \LdagNorm{z}^2\leq \frac{1}{\lambda_2}\norm{z}^2.
    \end{equation} 
    Therefore, since $\Qbf^{ij} = (\Qbf^{ij})^\top=(\Qbf^{ij})^2$, and using the fact that $\norm{z}_{\Lpbf}^2\leq\lambda_n\norm{z}^2$ for all $z\in\R^{nd}$, it follows  
    from \eqref{eq:proof:prop:ClosedConv:kLRate:UBquadTermInter}:
    \begin{align}
        \sum_{(i,j)\in\Ecal}p_{ij}\LdagNorm{\Qbf^{ij}\nabla f(x)}^2\nonumber
        &\leq \frac{1}{\lambda_2}\norm{\nabla f(x)-\nabla f(x^*)}_{\Lpbf}^2\\
        \label{eq:proof:prop:ClosedConv:kLRate:UBquadTerm}
        &\leq \frac{\lambda_n}{\lambda_2}\norm{\nabla f(x)-\nabla f(x^*)}^2.
    \end{align}
    
    \noindent We now analyze the third term of the right-hand side of \eqref{eq:proof:prop:ClosedConv:ExpandECk1}.
    From \eqref{eq:proof:prop:ClosedConv:kLRate:Qijnablaf(x^*)=0} we get
    \begin{align}
        \label{eq:proof:prop:ClosedConv:kLRate:Lpnablaf(x^*)=0}
        \Lpbf\nabla f(x^*) = \sum_{(i,j)\in\Ecal}p_{ij}\Qbf^{ij}\nabla f(x^*) = \mathbf{0}_{nd},
    \end{align}
    yielding \cite[Thm. 2.1.12]{nesterov2018lectures}:
    
    \vspace{-0.5cm}
    \begin{small}
        \begin{align*}
            \inProd{\nabla f(x)-\nabla f(x^*)}{x-x^*} \geq     \frac{\beta^{-1}\norm{\nabla f(x)-\nabla f(x^*)}^2}{1+\kappa^{-1}} + \frac{\alpha\norm{x-x^*}^2}{1+\kappa^{-1}}.
        \end{align*}
    \end{small}
    \vspace{-0.5cm}
    Hence, using the result above and \eqref{eq:proof:prop:ClosedConv:kLRate:LpdagNorm(z)<=Norm(z)/l2}, it follows that
    
    \begin{align}
        &-2h\inProd{\Lpbf\nabla f(x)}{\Lpbf^\dag(x-x^*)}\nonumber\\
        \label{eq:proof:prop:ClosedConv:kLRate:UBcrossTerm}
        &\leq -2h\frac{\beta^{-1}\norm{\nabla f(x)-\nabla f(x^*)}^2}{1+\kappa^{-1}} - 2h\frac{\alpha\lambda_2\LdagNorm{x-x^*}^2}{1+\kappa^{-1}}.
    \end{align}
    Injecting \eqref{eq:proof:prop:ClosedConv:kLRate:UBquadTerm} and \eqref{eq:proof:prop:ClosedConv:kLRate:UBcrossTerm} into \eqref{eq:proof:prop:ClosedConv:kLRate:Start} yields
    \begin{align}
    \label{eq:proof:prop:ClosedConv:kl:ALTERNATIVE}
        \E\brk{\LdagNorm{x^+-x^*}^2} 
        &\leq \prt{1-2h\frac{\alpha\lambda_2}{1+\kappa^{-1}}}\LdagNorm{x-x^*}^2\\
        &+\prt{h^2\frac{\lambda_n}{\lambda_2}-2h\frac{\beta^{-1}}{1+\kappa^{-1}}}\norm{\nabla f(x)-\nabla f(x^*)}^2.\nonumber
    \end{align}
    Observe that if $h\leq \frac{\lambda_2}{\lambda_n}\frac{2}{\alpha+\beta}$ then $h^2\frac{\lambda_n}{\lambda_2}-2h\frac{\beta^{-1}}{1+\kappa^{-1}}\leq0$.
    As a consequence, \eqref{eq:proof:prop:ClosedConv:kl:ALTERNATIVE} can be upper bounded using the definition of $\alpha$-strongly convex functions, and more specifically 
    \begin{align*}
        \norm{\nabla f(x)-\nabla f(x^*)}^2
        \geq \alpha^2\lambda_2\LdagNorm{x-x^*}^2,
    \end{align*}
    where we used the fact that $\norm{x-x^*}^2\geq\lambda_2\LdagNorm{x-x^*}^2$.
    It follows that 
    $$
    \E\brk{\LdagNorm{x^+-x^*}^2}\leq \LdagNorm{x-x^*}^2+(h^2\alpha^2\lambda_n-2h\alpha\lambda_2)\LdagNorm{x-x^*}^2,
    $$ 
    which concludes the proof.
\end{proof}

It is clear that the constant of convergence is less than one if $h\leq \frac{2\lambda_2}{\alpha\lambda_n}$, which is thus satisfied on all its range of validity since $h\leq \frac{2\lambda_2}{(\alpha+\beta)\lambda_n} \leq \frac{2\lambda_2}{\alpha\lambda_n}$.
We can then find the step-size which minimizes \eqref{eq:prop:ClosedConv:kLRate} and the corresponding constant of convergence.

\begin{cor}
\label{cor:ClosedConv:OptimalConvRate}
    The optimal constant of convergence in \eqref{eq:prop:ClosedConv:kLRate} under \eqref{eq:prop:ClosedConv:kLRate:h} is achieved for $h^* = \frac{2\lambda_2}{(\alpha+\beta)\lambda_n}$ which yields
    \begin{equation}
    \label{eq:cor:ClosedConv:OptimalConvRate}
        \E\brk{\LdagNorm{x^+-x^*}^2}\leq \prt{1-\frac{\lambda_2^2}{\lambda_n}\frac{1}{\kappa}}\LdagNorm{x-x^*}^2.
    \end{equation}
\end{cor}

Interestingly, Proposition~\ref{prop:ClosedConv:kLRate} shows that linear convergence can be achieved by the RCD algorithm with respect to the norm induced by $\Lpbf^\dag$ with a constant of convergence similar to that of classical algorithms based on gradient descent \cite{nesterov2018lectures,MAL-050}.

\begin{remark}[Complete graph]
    For the particular case of a complete graph with 1-dimensional homogeneous agents and uniform probabilities, the eigenvalues of $L_p$ are $\lambda_2=\lambda_n=\frac{1}{n-1}$ and the $\Lpbf^{\dagger}$-norm coincides with the Euclidean norm for all $z=x-y$, where $x,y\in\Sab$.
    Then, the result of Proposition~\ref{prop:ClosedConv:kLRate} becomes 
    $$\Ep{\norm{x^+-x^*}^2} \leq\prt{1-\frac{\alpha h}{n-1}(2-\alpha h)}\norm{x-x^*}^2.$$
    Since in that case by definition $h\leq\frac{2}{\alpha+\beta}\leq\frac1\alpha$, it follows that 
    $$\Ep{\norm{x^+-x^*}^2}\leq \prt{1-\frac{\alpha h}{n-1}}\norm{x-x^*}^2,$$
    which coincides with \cite[Eq. (13)]{monnoyer2021random}.
\end{remark}

\begin{remark}[Alternative rate]
\label{rem:ClosedConv:AltStepSize}
Starting from \eqref{eq:proof:prop:ClosedConv:kl:ALTERNATIVE} in the proof of Proposition~\ref{prop:ClosedConv:kLRate}, one can use a similar argument to derive the following alternative constant of convergence, valid for $\frac{\lambda_2}{\lambda_n}\frac{2}{\alpha+\beta}\leq h\leq \frac{\kappa^{-1}+\kappa_L}{\kappa^{-1}+1}\frac{\lambda_2}{\lambda_n^2\beta}$, with $\kappa_L=\frac{\lambda_n}{\lambda_2}$:
\begin{equation}
    \E\brk{\LdagNorm{x^+-x^*}^2} 
        \leq \prt{1-2\beta\lambda_2h+h^2\beta^2\frac{\lambda_n^2}{\lambda_2}}\LdagNorm{x-x^*}^2.    
\end{equation}
This result could be used in the rest of this work the same way as that of Proposition~\ref{prop:ClosedConv:kLRate} for the corresponding step-size.
This development is however omitted in this work.
\end{remark}

\subsection{Homogeneous agents and uniform probabilities}
\label{sec:ClosedConv:HomogeneousUniform}

\begin{assumption}\label{ass:homogeneous_agents}
The agents are homogeneous (\textit{i.e.}, $a = \mathds 1_n$) and the probabilities of selecting the edges during the implementation of the RCD algorithm are uniform (\textit{i.e.}, $p_{ij}=p$).
\end{assumption}

For the particular case of homogeneous agents and uniform probabilities, the matrix $L_p$ can be expressed as $L_p=\frac{p}{2}L$ where $L$ is the usual Laplacian matrix. 
In this case, the matrix $L_p$ can be associated to an electrical circuit \cite{dorfler2018electrical}, and we can use the concept of effective resistance to find an upper bound for the step size of the algorithm independently of $\lambda_2$.

Hence, the following proposition provides an alternative bound for the convergence of the RCD algorithm in the specific case described above, and can be used the same way as that of Proposition~\ref{prop:ClosedConv:kLRate} in the remainder of this work for that case.
However, for the sake of generality, we express the main result in the next section only in terms of Proposition~\ref{prop:ClosedConv:kLRate}.

\begin{proposition}
\label{prop:effective_resistance}
    Let a function $f(x):= \sum_{i=1}^n f_i(x_i)$ and $x^*:= \argmin_{x\in S_b} f(x)$.
    Under Assumptions~\ref{Ass:Statement:Fab} and \ref{ass:homogeneous_agents}, for any positive scalar
    \begin{equation}\label{eq:step_size_homogeneous}
    h\leq \frac{2p}{\lambda_n}\frac{2}{\alpha+\beta},    
    \end{equation}
    and for any initial point $x\in \Sab$, then the update rule \eqref{eq:ClosedConv:rule}
    applied on the randomly selected pair of agents $(i,j)\in \Ecal$ satisfies
    \vspace{-0.1cm}
    \begin{equation}
        \label{eq:prop:ClosedConv:kLRate_resistance}
        \E\brk{\LdagNorm{x^+-x^*}^2} 
        \leq \prt{1-2h\alpha\lambda_2+\frac{h^2\alpha^2\lambda_2\lambda_n}{2p}}\LdagNorm{x-x^*}^2.
    \end{equation}
\end{proposition}
\begin{proof}
    Since the matrices $\Qbf^{ij}$ are idempotent, the summation term of the second element in \eqref{eq:proof:prop:ClosedConv:ExpandECk1} can be expressed as:
    $$
    \LdagNorm{\Qbf^{ij}\nabla f(x)}^2=\inProd{\Qbf^{ij}\nabla f(x)}{\Qbf^{ij}\Lpbf^\dag \Qbf^{ij} \Qbf^{ij}\nabla f(x)}
    $$
    Then we can use an upper bound for the quadratic form and we obtain for each term:
    $$
    \LdagNorm{\Qbf^{ij}\nabla f(x)}^2\le\norm{\Qbf^{ij}\nabla f(x)}^2\lambda_{\max}(\Qbf^{ij}\Lpbf^\dag \Qbf^{ij}).
    $$
    Now, the matrix $\Qbf^{ij}\Lpbf^\dag \Qbf^{ij}$ is given by:
    $$
    \Qbf^{ij}\Lpbf^\dag \Qbf^{ij}=\prt{Q^{ij}L_p^\dag Q^{ij}}\otimes I_d,
    $$
    which implies that $\lambda_{\max}(\Qbf^{ij}\Lpbf^\dag \Qbf^{ij})=\lambda_{\max}(Q^{ij}L_p^\dag Q^{ij})$. Then, we have:
    $$
    Q^{ij}\invL Q^{ij}=\frac{1}{2}\prt{[L_P^\dag]_{ii}+[L_P^\dag]_{jj}-2[L_P^\dag]_{ij}}Q^{ij}=\frac{1}{2}r_{ij}Q^{ij},
    $$
    where $r_{ij}$ is the effective resistance between the agents $i$ and $j$. Since there is an edge between $i$ and $j$, we have $r_{ij}\le 1/p$. Then we have the following upper bound for the largest eigenvalue:
    \begin{equation}
    \label{eq:effective_resistance}
        \lambda_{\max}({Q^{ij}\invL Q^{ij}})\le \frac{1}{2p} \; \text{for all } (i,j)\in E,    
    \end{equation}
    and we get:
    \begin{align}
        \sum_{(i,j)\in\Ecal}p_{ij}\LdagNorm{\Qbf^{ij}\nabla f(x)}^2
        \nonumber
        &\leq \frac{1}{2p}\norm{\nabla f(x)-\nabla f(x^*)}_{\Lpbf}^2\\
        \label{eq:proof:prop:ClosedConv:kLRate:UBquadTerm2}
        &\leq \frac{\lambda_n}{2p}\norm{\nabla f(x)-\nabla f(x^*)}^2,
    \end{align}
    which replaces \eqref{eq:proof:prop:ClosedConv:kLRate:UBquadTerm}.
    The rest of the proof follows the same steps as in the proof of Proposition~\ref{prop:ClosedConv:kLRate}.
\end{proof}

Similarly to Proposition~\ref{prop:ClosedConv:kLRate}, the constant of convergence is strictly decreasing if $h\leq \frac{4p}{\alpha\lambda_n}$, which is 
always satisfied since $h\leq \frac{4p}{(\alpha+\beta)\lambda_n} \leq \frac{4p}{\alpha\lambda_n}$.
Hence, we can also find the optimal step-size for the algorithm, and the corresponding constant of convergence.

\begin{cor}
\label{cor:ClosedConv:OptimalConvRate:effectiveResistance}
    The optimal constant of convergence in \eqref{eq:prop:ClosedConv:kLRate_resistance} under \eqref{eq:step_size_homogeneous} is achieved for $h^* = \frac{4p}{(\alpha+\beta)\lambda_n}$ which yields
    \begin{equation}
    \label{eq:cor:ClosedConv:OptimalConvRate:effectiveResistance}
        \E\brk{\LdagNorm{x^+-x^*}^2}\leq \prt{1-\frac{2p\lambda_2}{\lambda_n}\frac{1}{\kappa}}\LdagNorm{x-x^*}^2.
    \end{equation}
\end{cor}

\begin{remark}
The upper bound for the step size derived in Proposition~\ref{prop:effective_resistance} is better suited for graphs with a small $\lambda_2$ (also known as \emph{algebraic connectivity}), that is, non-robust networks that can be easily disconnected \cite{jamakovic2008robustness}.
If we denote by $\mu$ the eigenvalues of $L$, which satisfy $\lambda=\frac{p}{2}\mu$, then we have that for Proposition~\ref{prop:ClosedConv:kLRate} the step size must satisfy $h\le \frac{\mu_2}{\mu_n}\frac{2}{\alpha+\beta}$ while for Proposition~\ref{prop:effective_resistance} the step size is upper bounded by $h\le \frac{4}{\mu_n}\frac{2}{\alpha+\beta}$. 
\end{remark}

\section{Convergence of RCD in open system}
\label{Sec:OpenSystems}

Let us define the ratio
\begin{equation}
    \label{eq:OpenSystems:rhoR}
    \rho_R := \frac{1-p_U}{p_U},
\end{equation}
which characterizes the expected number of replacements happening in the system between two consecutive RCD updates.
In particular, when $\rho_R\to0$, then the system converges to a closed system, and when $\rho_R\to\infty$, then replacements become so frequent that RCD updates are negligible.

In the following theorem we present the main result of this work, in which we derive the constant of convergence for the error achieved by the RCD algorithm in expectation in open system, under the form of an upper bound on that error.
The derivation of this result relies on the separate analysis of the effects of replacements and of RCD updates, which is enabled by Assumption~\ref{Ass:Statement:Indep}.

\begin{theorem}
\label{thm:OpenSystems:EConvRate_TAC}
    Let $M_{n,\kappa} = \frac{1}{\lambda_2}\bar M_{n,\kappa}$ with $\bar M_{n,\kappa}$ defined in \eqref{eq:Bound_Mnk}, and let $\bar\eta := \rho_R \frac{M_{n,\kappa}}{\alpha h(2\lambda_2-\lambda_n\alpha h)}$.
    In the setting described in Section~\ref{sec:Statement}, 
    the sequence of estimates $x^k$ generated by the iteration rule \eqref{eq:ClosedConv:rule} with $h\leq\frac{\lambda_2}{\lambda_n}\frac{2}{\alpha+\beta}$ satisfies for all $k$ and for any $\eta>\bar\eta$:
    \begin{align}
        \label{thm:OpenSystems:EConvRate}
        \Ep{\LdagNorm{x^{k+1}-\optx{k+1}}^2}-\Gamma_{\eta} \leq A_\eta\prt{\Ep{\LdagNorm{x^k-\optx{k}}^2}-\Gamma_\eta},
    \end{align}
    with
    \begin{align}
        \label{thm:OpenSystems:EConvRate:Aeta}
        A_\eta &:= 1-p_U\alpha h(2\lambda_2 - \alpha\lambda_nh) + (1-p_U)\frac{M_{n,\kappa}}{\eta};\\
        \label{thm:OpenSystems:EConvRate:Gammaeta}
        \Gamma_\eta &:=
        \frac{(1-p_U)M_{n,\kappa}(\eta+M_{n,\kappa})\eta}{p_U\eta\alpha h(2\lambda_2 - \alpha\lambda_nh) - (1-p_U)M_{n,\kappa}}.
    \end{align}
\end{theorem}
\begin{proof}
    Let us denote $ C^k := \LdagNorm{x^k-\optx{k}}^2$.
    From Assumption~\ref{Ass:Statement:Indep}, we have
    \begin{equation}
    \label{eq:proof:Thm:OpenSystems:ConvRate:Start}
        \Ep{C^{k+1}}
        = p_U\Ep{C^{k+1}|U} + (1-p_U)\Ep{C^{k+1}|R},
    \end{equation}
    where $U$ and $R$ respectively stand for the occurrence of an RCD update and a replacement event.
    Proposition~\ref{prop:ClosedConv:kLRate} then yields for $h\leq\frac{\lambda_2}{\lambda_n}\frac{2}{\alpha+\beta}$
    \begin{equation}
    \label{eq:proof:Thm:OpenSystems:ConvRate:U}
        \Ep{C^{k+1}|U} 
        \leq \prt{1 - 2\lambda_2\alpha h + \lambda_n\alpha^2h^2}\Ep{C^k}.
    \end{equation}
    Under a replacement event, we have
    $x^{k+1}=x^k$, and hence Proposition~\ref{prop:Repl:minChange} yields
    \begin{align}
        \Ep{C^{k+1}|R} 
        &= \E\left[\LdagNorm{x^{k+1}-x^{k,*}+x^{k,*}-x^{k+1,*}}^2\right]\nonumber\\
        &\leq \E\left[\prt{\LdagNorm{x^{k}-x^{k,*}}+\LdagNorm{x^{k,*}-x^{k+1,*}}}^2\right]\nonumber\\
        \label{eq:proof:Thm:OpenSystems:ConvRate:R}
        &= \E\left[\prt{\sqrt{C^k}+M_{n,\kappa}}^2\right].
    \end{align}
    Injecting \eqref{eq:proof:Thm:OpenSystems:ConvRate:U} and \eqref{eq:proof:Thm:OpenSystems:ConvRate:R} into \eqref{eq:proof:Thm:OpenSystems:ConvRate:Start} then yields the following nonlinear recurrence:
    \begin{align}
    \label{eq:proof:Thm:OpenSystems:ConvRate:nonlinrec}
        \Ep{C^{k+1}}
        &\leq \prt{1 - p_U\alpha h(2\lambda_2 - \lambda_n\alpha h)}\Ep{C^k}\nonumber\\
        &\ \ \ + (1-p_U)\prt{2M_{n,\kappa}\Ep{\sqrt{C^k}} + M_{n,\kappa}^2}.
    \end{align}
    Since $2x\leq \eta + \frac{x^2}{\eta}$ holds for all $x\geq0$ and $\eta>0$, we obtain 
    $\sqrt{C^k} \leq \frac\eta2 + \frac{C^k}{2\eta}$ for any $\eta>0$.
    Hence, it follows that
    \begin{align*}
        \Ep{C^{k+1}}
        &\leq \prt{1 - p_U\alpha h(2\lambda_2 - \lambda_n\alpha h) + (1-p_U)\frac{M_{n,\kappa}}{\eta}}\Ep{C^k}\nonumber\\
        &\ \ \ + (1-p_U)M_{n,\kappa}(\eta+M_{n,\kappa}).
    \end{align*}
    Observe that $1 - p_U\alpha h(2\lambda_2 - \lambda_n\alpha h) + (1-p_U)\frac{M_{n,\kappa}}{\eta}<1$ for $\eta>\bar\eta$, so that solving the linear recurrence yields the conclusion.
\end{proof}

Notice that the setting of Theorem~\ref{thm:OpenSystems:EConvRate_TAC} guarantees that the contraction rate satisfies $A_\eta<1$ (which is ensured for any $\eta>\bar\eta$) for any parametrization of the system (as long as $\rho_R<\infty$, \textit{i.e.}, if updates happen).
Hence, $\Gamma_\eta$ provides an upper bound on the asymptotic expected error, and a few algebraic manipulations yield
\begin{equation}
    \label{eq:cor:OpenSystems:etabar:limGamma}
    \lim\sup_{k\to\infty}\Ep{\LdagNorm{x^{k}-\optx{k}}^2}\leq 
    \Gamma_\eta
    = \frac{\bar\eta(M_{n,\kappa}+\eta)}{1-\bar\eta/\eta}.
\end{equation}

When $\eta\to\infty$, then the contraction rate becomes minimal, \textit{i.e.}, $A_\eta \to 1-p_U\alpha h(2\lambda_2-\lambda_n\alpha h)$, and the asymptotic error diverges. 
Observe that for $\eta>\bar\eta$, then $\Gamma_\eta$ is convex, and one can determine the value of $\eta$ that minimizes the upper bound on the asymptotic expected error $\Gamma_\eta$, as presented in the following corollary.

\begin{cor}
\label{cor:OpenSystems:etastar}
    When $\eta = \eta^* = \bar\eta\prt{1+\sqrt{1+\frac{M_{n,\kappa}}{\bar\eta}}}$, the convergence of the RCD algorithm in open system is guaranteed with minimal upper bound on the asymptotic error $\Gamma_{\eta^\ast}$, and: 
    \begin{align}
        \label{eq:cor:OpenSystems:etastar:Aetastar}
        A_{\eta^*}
        &= 1 - p_U\alpha h(2\lambda_2-\lambda_n\alpha h)\frac{\sqrt{1+\frac{M_{n,\kappa}}{\bar\eta}}}{1+\sqrt{1+\frac{M_{n,\kappa}}{\bar\eta}}};\\
        \label{eq:cor:OpenSystems:etastar:Gammaetastar}
        \Gamma_{\eta^*}
        &= (\eta^\ast)^2
        = \bar\eta^2
        \prt{1+\sqrt{1+\frac{M_{n,\kappa}}{\bar\eta}}}^2,
    \end{align}
    where we recall that $\bar\eta:=\rho_R\frac{M_{n,\kappa}}{\alpha h(2\lambda_2-\lambda_n\alpha h)}$.
\end{cor}
\begin{proof}
    Observe that $\Gamma_\eta$ is convex for $\eta>\bar\eta$, and we have 
    $\frac{d}{d\eta}\Gamma\rvert_{\eta=\eta^*}=0$ with $\eta^* = \arg\min_{\eta>\bar\eta}\Gamma_\eta$.
    Hence we compute
    \begin{align}
    \label{eq:proof:cor:OpenSystems:etastar:dGammadeta}
        \frac{d}{d\eta}\Gamma_\eta
        &= \frac{\bar\eta}{(\eta-\bar\eta)^2}(\eta^2-2\bar\eta\eta-M_{n,\kappa}\bar\eta)=0,
    \end{align}
    which is satisfied for
    \begin{align*}
        &\eta_1^* = \bar\eta+\sqrt{\bar\eta^2+M_{n,\kappa}\bar\eta}&
        &;&
        &\eta_2^* = \bar\eta-\sqrt{\bar\eta^2+M_{n,\kappa}\bar\eta}.
    \end{align*}
    Since $\eta_2^*<\bar\eta$, it must be rejected, and it follows that 
    \begin{equation*}
        \eta^* 
        = \bar\eta+\sqrt{\bar\eta^2+M_{n,\kappa}\bar\eta}
        = \bar\eta\prt{1+\sqrt{1+\frac{M_{n,\kappa}}{\bar\eta}}}.
    \end{equation*}
    We can then compute
    \begin{align*}
        A_{\eta^*}
        &= 1 - p_U\alpha h(2\lambda_2-\lambda_n\alpha h) + \frac{(1-p_U)M_{n,\kappa}}{\bar\eta\prt{1+\sqrt{1+M_{n,\kappa}/\bar\eta}}},
    \end{align*}
    and a few algebraic manipulations yield \eqref{eq:cor:OpenSystems:etastar:Aetastar}.
    Now observe that $\Gamma_{\eta}=\eta^2$ if and only if
    \begin{align*}
        \eta^2-2\bar\eta\eta-M_{n,\kappa}\bar\eta = 0,
    \end{align*}
    which is equivalent to \eqref{eq:proof:cor:OpenSystems:etastar:dGammadeta} for $\eta>\bar\eta$, so that the solution is $\eta^\ast$. 
    Hence, $\Gamma_{\eta^\ast}=\prt{\eta^\ast}^2$, which yields \eqref{eq:cor:OpenSystems:etastar:Gammaetastar}.
\end{proof}

\begin{remark}
    The bound $\Gamma_\eta$ in \eqref{eq:cor:OpenSystems:etabar:limGamma} depends on $M_{n,\kappa}$ and $\bar\eta$, which itself depends on $M_{n,\kappa}$ and other parameters as well.
    In fact, $\bar\eta$ can be interpreted as the ratio between the spurious effect of replacements and the advantageous effect of RCD steps on the error.
    Briefly, this means that the greater the effect of replacements and the smaller that of updates (as discussed in the previous sections), then the larger the asymptotic error.
\end{remark}

Observe that Theorem~\ref{thm:OpenSystems:EConvRate_TAC} and Corollary~\ref{cor:OpenSystems:etastar} provide \emph{upper bounds} on the expected error of the algorithm, and therefore induce a certain conservatism with respect to the actual error.
This essentially follows from Proposition~\ref{prop:Repl:minChange} whose aim is to bound the additive error injected at one single replacement, whereas tighter bounds might be derived on the sum of those additive errors.
This is especially true as $\rho_R$ grows.

When no replacements happen, \textit{i.e.}, $\rho_R\to0$, then the system behaves as a closed system, and we retrieve the corresponding convergence behavior: the expected asymptotic error $\Gamma_\eta\to0$ and the contraction rate $A_\eta\to1-\alpha h(2\lambda_2-\lambda_n\alpha h)$ for all $\eta>0$, consistently with the constant of convergence of the RCD in closed system derived in \eqref{eq:prop:ClosedConv:kLRate}.
By contrast, as $\rho_R$ gets larger, \textit{i.e.}, as replacements become more frequent, then the expected asymptotic error increases, and the contraction rate $A_{\eta^*}$ gets closer to $1$ (observe that $A_{\eta^*}<1$ remains true as long as $\rho_R<\infty$).
In the particular limit case where $\rho_R\to\infty$, then the minimal upper bound on the expected asymptotic error becomes $\Gamma_{\eta^*}\to4\bar\eta\to\infty$ and $A_{\eta^*}\to1$.

Interestingly, within the allowed range of $h$, $\Gamma_{\eta^*}$ decays as the step-size $h$ increases, suggesting that choosing $h$ as large as possible leads to the smallest value for the upper bound on the expected asymptotic error $\Gamma_{\eta^*}$.
This means that the only limitation on the choice of the step-size comes from the analysis of the algorithm in closed system ($h\leq\frac{\lambda_2}{\lambda_n}\frac{2}{\alpha+\beta}$ from Proposition~\ref{prop:ClosedConv:kLRate} in our case), and that no particular precaution should be taken regarding the open character of the system.

\begin{remark}
\label{rem:OpenSystems:OtherAlgos}
    The methodology we used in this section can easily be extended to other algorithms than the RCD algorithm.
    In particular the results of Theorem~\ref{thm:OpenSystems:EConvRate_TAC} and Corollary \ref{cor:OpenSystems:etastar} can be adapted to any algorithm with linear convergence in closed systems, that is, such that 
    \begin{equation}
        \label{eq:rem:OpenSystems:OtherAlgos:ClosedConv}
        \LdagNorm{x^+-x^*}^2\leq K\LdagNorm{x-x^*}^2,
    \end{equation}
    with some positive $K<1$.
    In that case, the same constant of convergence as that presented in Theorem~\ref{thm:OpenSystems:EConvRate_TAC} is obtained with 
    \begin{align}
        \label{rem:OpenSystems:EConvRate:OtherAlgos:Aeta}
        A_\eta &:= 1-p_U(1-K)+(1-p_U)\frac{M_{n,\kappa}}{\eta};\\
        \label{rem:OpenSystems:EConvRate:OtherAlgos:Gammaeta}
        \Gamma_\eta &:= \frac{(1-p_U)M_{n,\kappa}(\eta+M_{n,\kappa})\eta}{p_U\eta(1-K)-(1-p_U)M_{n,\kappa}}.
    \end{align}
    We can show that convergence can be guaranteed in open system following a similar argument as that used to prove Theorem~\ref{thm:OpenSystems:EConvRate_TAC} if $K<1$.
    Hence, this analysis can be applied \textit{e.g.}, to the results presented in Proposition~\ref{prop:effective_resistance} or in Remark~\ref{rem:ClosedConv:AltStepSize}.
\end{remark}

To illustrate the results of Theorem~\ref{thm:OpenSystems:EConvRate_TAC}, we consider systems with piecewise quadratic local cost functions $f_i$: for $\varphi_{i1},\varphi_{i2}\in\brk{\frac{\alpha}{2},\frac{\beta}{2}}$, the cost function $f_i$ is given by
\begin{equation}\label{eq:piecewise_quadratic}
f_i(x_i) = \begin{cases}
  \varphi_{i1}(x_i-\nu_i)^2,  & \text{if }x_i< \nu_i \\
  \varphi_{i2}(x_i-\nu_i)^2, & \text{if }x_i\ge\nu_i
\end{cases},    
\end{equation}
where $\nu_i$ is the minimizer of $f_i$ satisfying Assumption~\ref{Ass:Statement:B(0,c)&f*=0}.
Such function therefore satisfies Assumption~\ref{Ass:Statement:Fab} as well.
Observe that no assumption on the way we choose the local cost function $f_i$ of a joining agent at replacements is required in the derivation of our results.
Hence, we consider two possible cases for that choice:
\emph{random}, where the parameters $\varphi_{i1}$ and $\varphi_{i2}$ are uniformly randomly chosen in $\brk{\frac{\alpha}{2},\frac{\beta}{2}}$, and \emph{adversarial}, where these parameters are arbitrarily chosen to maximize the error $\LdagNorm{x^k-\optx{k}}^2$ after the replacement among $100$ realizations of such uniform random choice.

\begin{figure}[h!]
    \centering
    \includegraphics[width=0.5\textwidth,clip = true, trim=1cm 10cm 1cm 10cm,keepaspectratio]{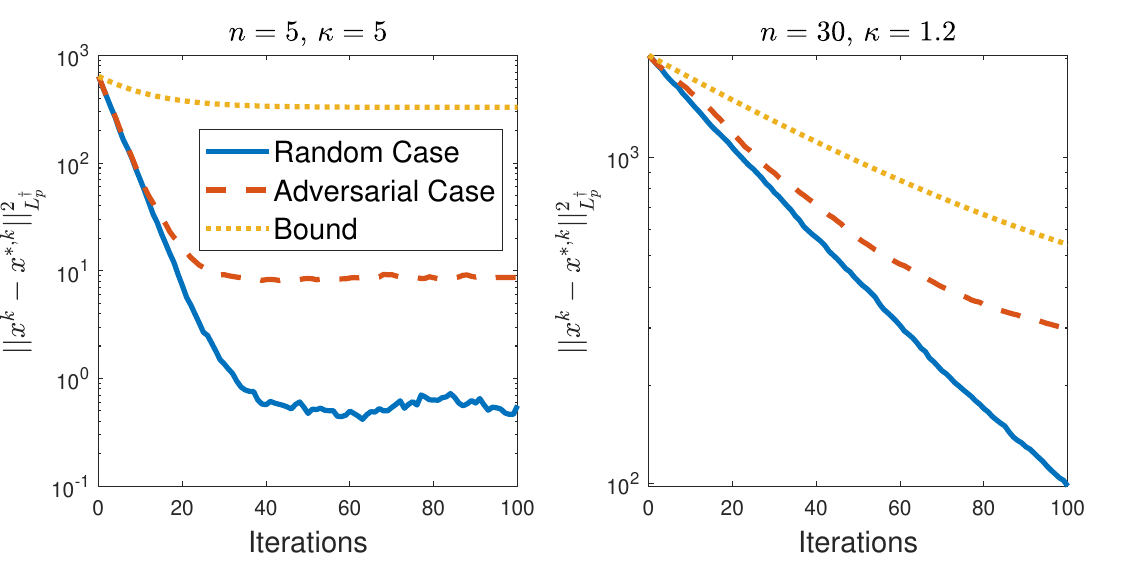}
    \caption{Performance of the RCD algorithm in a complete graph constituted of respectively $n=5$ agents with $\kappa=5$ (left) and $n=30$ agents with $\kappa=1.2$ (right), with $p_U=0.95$ and $b=1$, and where each local objective function is defined by \eqref{eq:piecewise_quadratic}.
    The plain blue and red dashed lines represent the actual performance of the algorithm averaged over 500 realizations of the process, respectively for the random and adversarial replacements cases.
    The yellow dotted line is the upper bound \eqref{thm:OpenSystems:EConvRate} obtained from Corollary~\ref{cor:OpenSystems:etastar}.}
    \label{fig:Complete:n5Kappa5_n30Kappa1.2}
\end{figure}

In Fig.~\ref{fig:Complete:n5Kappa5_n30Kappa1.2}, we show the evolution of the expected error $\Ep{\LdagNorm{x^k-\optx{k}}^2}$ simulated for a network with interconnections defined by a complete graph, homogeneous agents and uniform probabilities. We consider two parametrizations of $\kappa$ and $n$ in both random and adversarial replacement cases, and we compare the simulations with \eqref{thm:OpenSystems:EConvRate} using the values given by Corollary~\ref{cor:OpenSystems:etastar}.
The figure shows that convergence is indeed guaranteed for the RCD in the presented settings and that the result of Corollary~\ref{cor:OpenSystems:etastar} shows some conservatism, which is inherited from Proposition~\ref{prop:Repl:minChange}.
It is interesting to point out that these settings respectively make use of $\bar M_{n,\kappa}^2=\theta_{n,\kappa}$ for $n=30$, $\kappa=1.2$, and of $\bar M_{n,\kappa}^2=\psi_{n,\kappa}$ for $n=5$, $\kappa=5$, consistently with the description of $\bar M^2_{n,\kappa}$ in the homogeneous case of Section~\ref{sec:Replacements}. 
This highlights the impact of those parameters in the tightness of the bound $\bar M^2_{n,\kappa}$ used to derive our main results.

\begin{figure}[h!]
    \centering
    \includegraphics[width=0.5\textwidth,clip = true, trim=1cm 10cm 1cm 10cm,keepaspectratio]{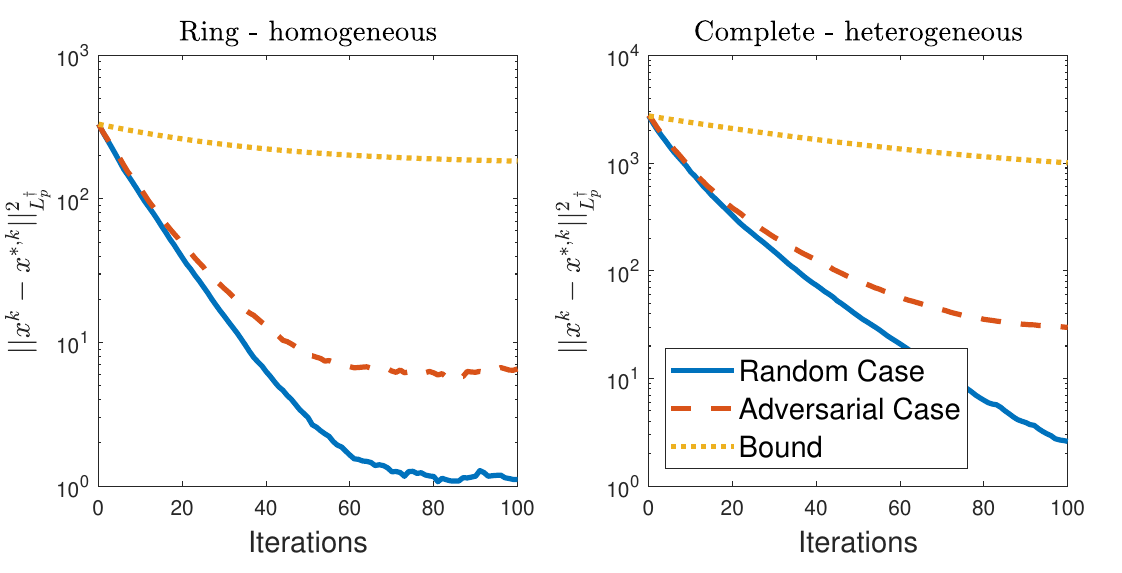}
    \caption{Performance of the RCD algorithm with $n=5$ agents, $\kappa=1.2$, $p_U=0.95$ and $b=1$, respectively in (left) a ring graph with homogeneous agents (\textit{i.e.}, $a_i=1$ for all $i$) and (right) a complete graph with heterogeneous agents (\textit{i.e.}, $a_1=10$, $a_i=1$ for $i>1$), and where each local objective function is defined by \eqref{eq:piecewise_quadratic}. 
    The plain blue and red dashed lines represent the actual performance of the algorithm averaged over 500 realizations of the process, respectively for the random and adversarial replacements cases.
    The yellow dotted line is the upper bound \eqref{thm:OpenSystems:EConvRate} obtained from Corollary~\ref{cor:OpenSystems:etastar}.}
    \label{fig:HoRing&HeComplete:n5Kappa1.2}
\end{figure}

In Fig.~\ref{fig:HoRing&HeComplete:n5Kappa1.2}, we compare the simulated performance in both replacement cases with the upper bound from Corollary~\ref{cor:OpenSystems:etastar} for a ring graph with homogeneous agents and a complete graph with heterogeneous agents.
By contrast with the previous illustrations, the ring graph setting implies a different, sparse, topology which thus reduces the range of validity for the step-size $h$ due to the small value of $\lambda_2$ that does not scale with $n$.
Similarly, the heterogeneous setting impacts $\lambda_2$ and consequently reduces the range of $h$, due to the imbalance in $\Lpbf$.
Those moreover affect the behavior of the norm $\LdagNorm{\cdot}^2$.
Furthermore, the heterogeneous setting influences $\bar M_{n,\kappa}$ as well, and hence increases the effect of replacements on the bounds.
Nevertheless, even though they differ quantitatively, the results of Fig.~\ref{fig:HoRing&HeComplete:n5Kappa1.2} are qualitatively similar to the case of the complete graph with homogeneous agents presented in Fig.~\ref{fig:Complete:n5Kappa5_n30Kappa1.2}.

\section{Conclusion}
\label{Sec:conclusions}
We have studied the behavior of the distance to the minimizer for the resource allocation problem and proved linear convergence of the Random Coordinate Descent algorithm in an appropriate norm for the closed system. 
We analyzed the algorithm for a general graph topology and possible heterogeneous agents in an open multi-agent systems scenario when agents can be replaced during the iterations. 
Under replacement events, we proved that for an appropriate step-size, the algorithm cannot converge to the instantaneous minimizer due to the perturbations generated by the replacements but is stable. 
We derived an upper bound  for the error in expectation which depends on the variation of the minimizer due to replacements and the frequency of these events.

A natural continuation of this work would be to handle the case where the budget and weights in the constraint can vary in time. 
Also, it would be interesting to consider more general equality constraints between the states of the agents and block updates at each iteration \cite{richtarik2014iteration}, such that the optimization is performed along more than one edge. Finally, a possible varying size of the system is an interesting direction for future work, where agents could join and leave the network independently of the current state of the system.
This extension would however introduce a significant amount of new challenges, especially regarding the impact of (dis)connections of agents in terms of both graph properties and vector dimensions in our analysis.

\bibliographystyle{IEEEtran}
\bibliography{TAC_Open_RCD.bib}

\appendix


\subsection{Proof of Lemma~\ref{lem:Repl:minLoc:Global}}\label{appendix:minLocLemmas1}
Let us consider some $x\in\Sab$ such that $x\notin B(0,R_{b,\kappa})$, and let $\bar x^* = \argmin_{x\in\R^{nd}}f(x)$ denote the global minimizer without constraint.
    We have $\norm{x}>R_{b,\kappa}$ by definition and $\norm{\bar x^*}\leq \sqrt{n}c$ since $\bar x^* \in B(\mathbf{0}_d,c)^n$, and it thus follows that $\norm{x-\bar x^*}>R_{b,\kappa}-\sqrt{n}c$.
    Hence, since $f$ is $\alpha$-strongly convex, and since  
    $f(\bar x^*)=0$ from Assumption~\ref{Ass:Statement:B(0,c)&f*=0}, we have
    \begin{equation*}
        f(x)\! 
        \geq\! \frac\alpha2\norm{x-\bar x^*}^2
        >\! \frac\alpha2\kappa\!\prt{\!\!\sqrt{n}c+\frac{\norm{b}}{\norm{a}}}^2\!\!
        =\! \frac{\beta }{2}\!\prt{\!\!\sqrt{n}c+\frac{\norm{b}}{\norm{a}}}^2.
    \end{equation*}
    Now let us define $x_b := \frac{1}{\norm{a}^2}(a\otimes I_d)b$, which is a feasible point with norm $\norm{x_b}=\frac{\norm{b}}{\norm{a}}$.
    Since $f$ is $\beta$-smooth, and since $f(\bar x^*)=0$ from Assumption~\ref{Ass:Statement:B(0,c)&f*=0},
    we get
    \begin{align*}
        f(x_b) \leq \frac{\beta}{2}\norm{x_b-\bar x^*}^2 
        \leq \frac{\beta }{2} \prt{\sqrt{n}c+\frac{\norm{b}}{\norm{a}}}^2.
    \end{align*}
    Finally, since $x_b\in \Sab$, we have $f(x^*)\leq f(x_b)$ by definition.
    Combining all the inequalities above then yields
    \begin{align*}
        f(x) > \frac{\beta }{2}\prt{\sqrt{n}c+\frac{\norm{b}}{\norm{a}}}^2 
        \geq f(x_b) 
        \geq f(x^*),
    \end{align*}
    which implies that $x$ cannot be the minimizer of the problem and concludes the proof.
    
\subsection{Proof of Lemma~ \ref{lem:Repl:minLoc:Local}}
\label{appendix:minLocLemmas2}
    Let us denote $\bar x_i^*$ the minimizer of $f_i$ without constraint which satisfies $f_i(\bar x_i^*)=0$ and $\nabla f_i(\bar x_i^*)=\mathbf 0_d$. From $\beta$-smoothness of the local functions we have \cite[Thm. 2.1.5]{nesterov2018lectures}:
    $$
    \norm{\nabla f_i(x_i^*)}^2\le\beta\inProd{\nabla f_i(x_i^*)}{x_i^*-\bar x_i^*}.
    $$
    Then, from the optimality condition \eqref{eq:Statement:Optimality}: 
    $$
    a_i^2\norm{\lambda^*}^2\le\beta\inProd{\lambda^*}{a_i(x_i^*-\bar x_i^*)}.
    $$
    By summing over all the $i$, we obtain:
    $$
    \norm{a}^2\norm{\lambda^*}^2\le \beta\inProd{\lambda^*}{\sum_{i=1}^n a_i\prt{x_i^*-\bar x_i^*}}.
    $$
    We use the Cauchy-Schwarz inequality to get:
    $$
    \norm{a}^2\norm{\lambda^*}^2\le \beta\norm{\lambda^*}\norm{\sum_{i=1}^n a_i\prt{x_i^*-\bar x_i^*}},
    $$
    and by using the triangle inequality and the fact that $\sum_{i=1}^n a_ix_i^*=b$ we obtain:
    \begin{align}
    \norm{\lambda^*}&\le\frac{\beta}{\norm{a}^2}\prt{\norm{b}+\norm{\sum_{i=1}^na_i\bar x_i^*}}\nonumber\\
    &\le\frac{\beta}{\norm{a}^2}\prt{\norm{b}+c\norm{a}_1}\nonumber,  
    \end{align}
    which corresponds to \eqref{eq:equation2_proof_bound_minimizer}.
    From $\alpha$-strongly convexity of the local functions we have:
    $$
    \alpha \norm{x_i-\bar x_i^*}\le \norm{\nabla f(x_i)}.
    $$
    By using the reverse triangle inequality and the optimality condition we get
    \begin{equation}
    \label{eq:equation1_proof_bound_minimizer}
        \norm{x_i^*}\le\frac{1}{\alpha}\norm{a_i\lambda^*}+\norm{\bar x_i^*}\le \frac{a_i}{\alpha}\norm{\lambda^*}+c.    
    \end{equation}
    Equation \eqref{eq:lem:Repl:minLoc:Local:nolambda} then follows from combining \eqref{eq:equation2_proof_bound_minimizer} and \eqref{eq:equation1_proof_bound_minimizer}.

\subsection{Proof of Proposition~\ref{prop:Repl:minChange}}
\label{sec:Appendix:MinChange:Proof}
\paragraph{Proof that $\psi_{n,\kappa}$ is an upper bound}
Observe that $x^{(1)},x^{(2)}\in B(\mathbf 0_{nd},R_{b,\kappa})$ from Lemma~\ref{lem:Repl:minLoc:Global}, so that: 
\begin{align*}
    \norm{x^{(1)}-x^{(2)}}^2 
    \leq 2\prt{\norm{x^{(1)}}^2+\norm{x^{(2)}}^2} 
    \leq 4R_{b,\kappa}^2,
\end{align*}
which yields that $\psi_{n,\kappa}$ is an upper bound.
\paragraph{Proof that $\chi_{n,\kappa}$ is an upper bound} We remind the reader that for $i=1,\ldots,n-1$, we have
$\nabla f_i(x_i^{(q)}) = a_i \lambda^{(q)}$, with $q=1,2$.
From $\alpha$-strongly convexity of the local functions, it follows
that for all $i=1,\ldots,n-1$:
$$
a_i\inProd{\lambda^{(1)}-\lambda^{(2)}}{x_i^{(1)}-x_i^{(2)}} \geq  \alpha\norm{x_i^{(1)}-x_i^{(2)}}^2.
$$
Let us define $y^{(q)}\in\R^{d(n-1)}$ the vector such that 
$y_i^{(q)} = x_i^{(q)}$
for $q=1,2$ and for $i=1,\ldots,n-1$.
Using the fact that $\sum_{i=1}^n a_ix_i^{(q)} = b$ for $q=1,2$ and summing up the above inequalities over all $i=1,\ldots,n-1$ yields
$$
a_n\inProd{\lambda^{(1)}-\lambda^{(2)}}{x_n^{(2)}-x_n^{(1)}} \geq   \alpha\norm{y^{(1)}-y^{(2)}}^2,
$$
where we used the fact that $\sum_{i=1}^{n-1}a_ix_i^{(q)}+a_nx_n^{(q)}=b$.
By using Cauchy-Schwarz inequality and triangle inequality we obtain
\begin{align}
    \norm{y^{(1)}-y^{(2)}}^2&\le \frac{a_n}{\alpha}\prt{\norm{\lambda^{(1)}}+\norm{\lambda^{(2)}}}\prt{\norm{x_n^{(1)}}+\norm{x_n^{(2)}}}\nonumber
\end{align}
Then, we use \eqref{eq:equation2_proof_bound_minimizer} and \eqref{eq:lem:Repl:minLoc:Local:nolambda} to get
\begin{equation}\label{eq:chi_eq1}     
    \norm{y^{(1)}-y^{(2)}}^2\le 4\prt{\frac{a_n}{\norm{a}^2}\kappa\prt{\norm{b}+c\norm{a}_1}+c}.
\end{equation}
By definition we have
\begin{equation}\label{eq:norm_x1_x2}
    \norm{x^{(1)}-x^{(2)}}^2=\norm{y^{(1)}-y^{(2)}}^2+\norm{x_n^{(1)}-x_n^{(2)}}^2.
\end{equation}
We apply triangle inequality and \eqref{eq:lem:Repl:minLoc:Local:nolambda} to obtain
\begin{align}
    \norm{x^{(1)}-x^{(2)}}^2&\le\norm{y^{(1)}-y^{(2)}}^2+\prt{\norm{x_n^{(1)}}+\norm{x_n^{(2)}}}^2\nonumber\\
    &\le\norm{y^{(1)}-y^{(2)}}^2+2\prt{\norm{x_n^{(1)}}^2+\norm{x_n^{(2)}}^2}\nonumber\\
    &\le \norm{y^{(1)}-y^{(2)}}^2+4\prt{\frac{a_n}{\norm{a}^2}\kappa\prt{\norm{b}+c\norm{a}_1}+c}\label{eq:chi_eq2}
\end{align}
Finally, the result \eqref{eq:prop:Repl:minChange:chi} yields by combining \eqref{eq:chi_eq1} and \eqref{eq:chi_eq2} and using the fact that $a_n\le a_+$.
\paragraph{Proof that $\theta_{n,\kappa}$ is an upper bound}
Since the functions are $\alpha$-strongly convex and $\beta$-smooth, we have
that for all $i=1,\ldots,n-1$:
\begin{multline*}
    a_i(1+\kappa^{-1})\inProd{\lambda^{(1)}-\lambda^{(2)}}{x_i^{(1)}-x_i^{(2)}} \geq \\ \beta^{-1}a_i^2\norm{\lambda^{(1)}-\lambda^{(2)}}^2 + \alpha\norm{x_i^{(1)}-x_i^{(2)}}^2.
\end{multline*}
By summing up the above inequalities over all $i=1,\ldots,n-1$ yields
\begin{multline*}
    a_n(1+\kappa^{-1})\inProd{\lambda^{(1)}-\lambda^{(2)}}{x_n^{(2)}-x_n^{(1)}} \geq \\ m\beta^{-1}\norm{\lambda^{(1)}-\lambda^{(2)}}^2 + \alpha\norm{y^{(1)}-y^{(2)}}^2,
\end{multline*}
where $m=\sum_{i=1}^{n-1} a_i^2$.
By using Cauchy-Schwarz inequality we obtain:
\begin{multline*}
    a_n(1+\kappa^{-1})\norm{\lambda^{(1)}-\lambda^{(2)}} \norm{x_n^{(2)}-x_n^{(1)}} \geq \\ m\beta^{-1}\norm{\lambda^{(1)}-\lambda^{(2)}}^2 + \alpha\norm{y^{(1)}-y^{(2)}}^2.
\end{multline*}
This can be written as follows
\begin{equation}
\label{eq:proof:prop:Repl:monChange:equation_proof_difference_minimizers}
    \alpha\norm{y^{(1)}-y^{(2)}}^2 \leq \phi(\norm{\lambda^{(1)}-\lambda^{(2)}}),
\end{equation}
where 
\begin{equation*}
    \phi(z) = -m\beta^{-1}z^2 + a_n(1+\kappa^{-1})\norm{x_n^{(2)}-x_n^{(1)}}z.
\end{equation*}
Since $\phi$ is a concave parabola, there exists $\phi^* = \max_z\phi(z)<\infty$ such that $\phi(z)\leq \phi^*$ for all $z$ given by
\begin{equation*}
    \phi^\ast 
    = \frac{a_n^2(1+\kappa^{-1})^2\norm{x_n^{(2)}-x_n^{(1)}}^2}{4m\beta^{-1}},
\end{equation*}
and it follows by using \eqref{eq:proof:prop:Repl:monChange:equation_proof_difference_minimizers} that
\begin{align*}
    \norm{y^{(1)}-y^{(2)}}^2
    &\leq \frac{a_n^2(1+\kappa^{-1})^2}{\alpha\beta^{-1}} \frac{\norm{x_n^{(1)}-x_n^{(2)}}^2}{4m}\\
    &= \prt{\sqrt\kappa + \frac{1}{\sqrt \kappa}}^2 \frac{a_{+}^2\norm{x_n^{(1)}-x_n^{(2)}}^2}{4\prt{\norm{a}^2-a_+^2}}.
\end{align*}
Equation \eqref{eq:prop:Repl:minChange:theta} then follows from \eqref{eq:norm_x1_x2} using the fact that $x_n^{(1)}$ and $x_n^{(2)}$ are upper bounded by \eqref{eq:lem:Repl:minLoc:Local:nolambda}, which thus yields $\theta_{n,\kappa}$, and the conclusion follows.

\begin{IEEEbiography}[{\includegraphics[width=1in,height=1.25in,clip,keepaspectratio]{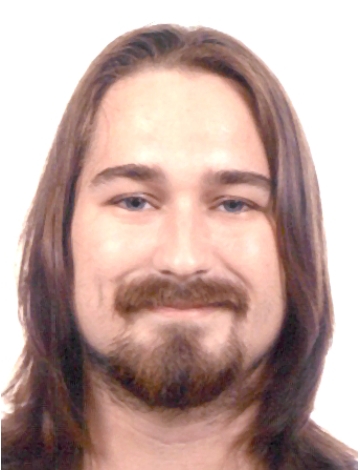}}]{Charles Monnoyer de Galland}
is a postdoctoral researcher at UCLouvain in the ICTEAM Institute.

He obtained an engineering degree in applied mathematics (2018), and the PhD degree in mathematical engineering (2022) as a FRIA fellow (F.R.S.-FNRS) from the same university.
His research interests are centered around the analysis of open multi-agent systems and decentralized optimization.
He was the recipient of the Networks and Communication Systems TC Outstanding Student Paper Prize of the IEEE Control Systems Society in 2022.
\end{IEEEbiography}

\begin{IEEEbiography}[{\includegraphics[width=1in,height=1.25in,clip,keepaspectratio]{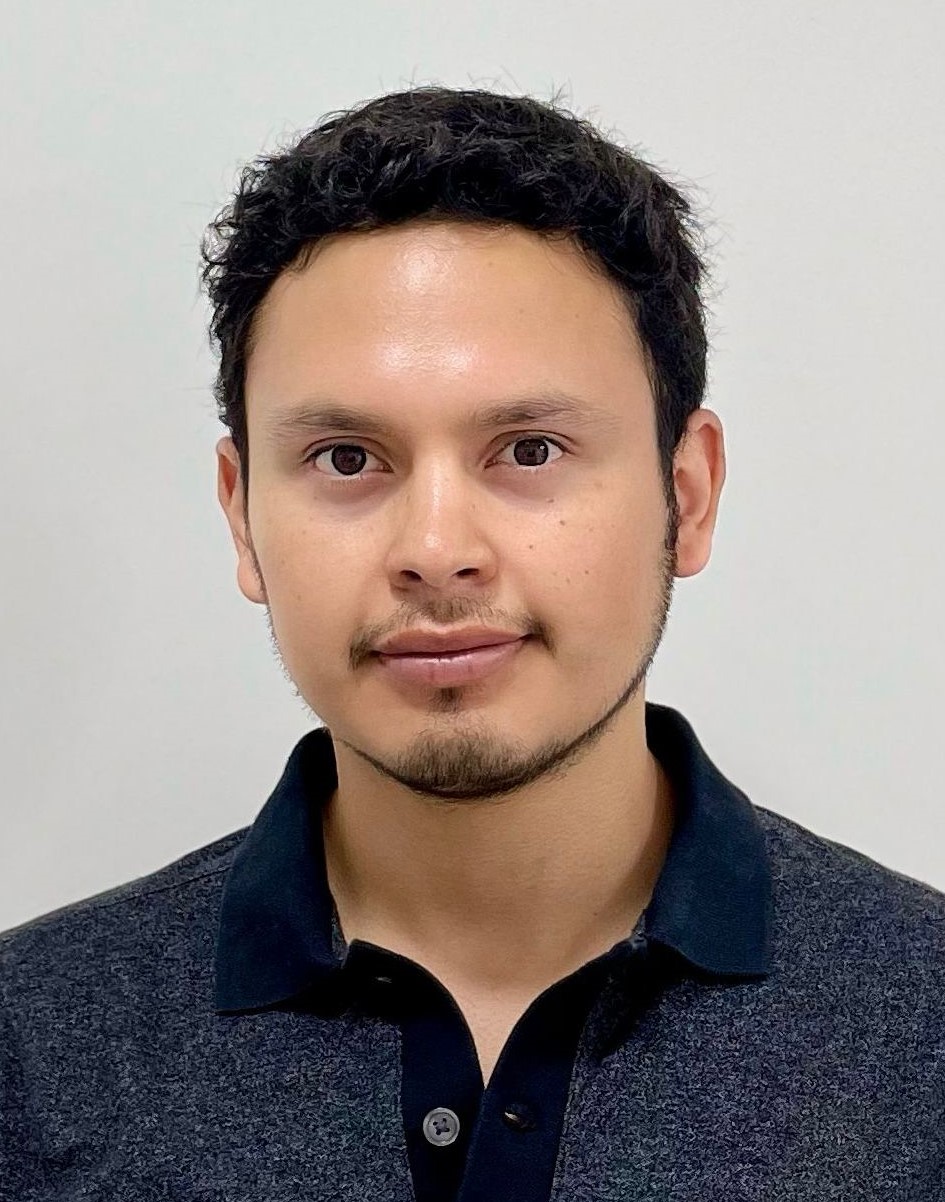}}]{Renato Vizuete}
received the M.S. degree (très bien) in Systems, Control, and Information Technologies from Université Grenoble Alpes, France (2019), and the PhD degree in Automatic Control from Université Paris-Saclay, France (2022). He is currently a postdoctoral researcher at UCLouvain, Belgium, in the ICTEAM Institute, holding a postdoctoral scholarship from the same university. His research interests include  multi-agent systems, distributed optimization, control theory and networked control systems. He was the recipient of the Networks and Communication Systems TC Outstanding Student Paper Prize of the IEEE Control Systems Society in 2022, and the Second Thesis Prize in the category Impact Science of the Fondation CentraleSupélec in 2023.
\end{IEEEbiography}

\begin{IEEEbiography}[{\includegraphics[width=1in,height=1.25in,clip,keepaspectratio]{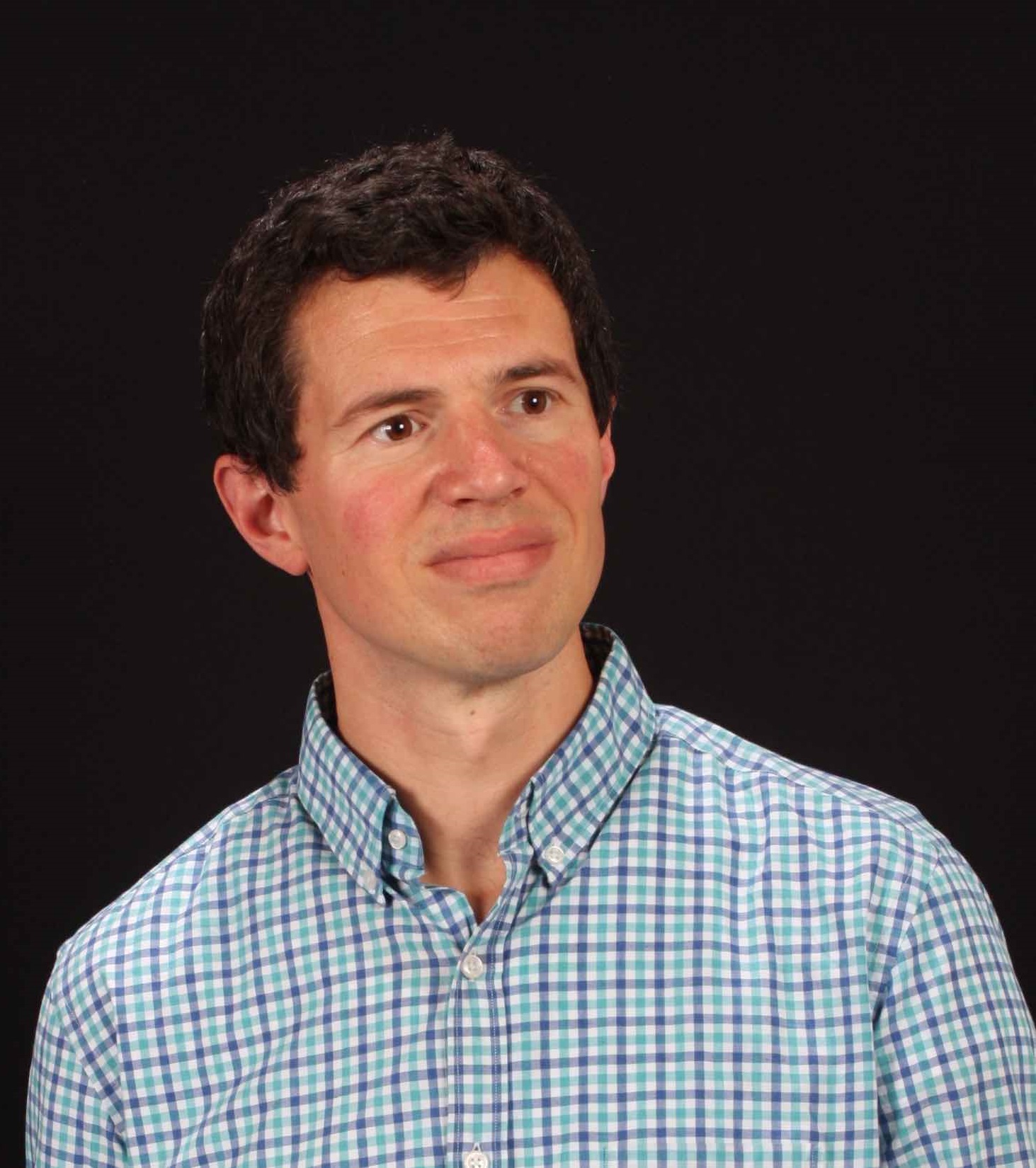}}]{Julien M. Hendrickx}
is professor of mathematical engineering at UCLouvain, in the Ecole Polytechnique de Louvain since 2010.

He obtained an engineering degree in applied mathematics (2004) and a PhD in mathematical engineering (2008) from the same university. He has been a visiting researcher at the University of Illinois at Urbana Champaign in 2003-2004, at the National ICT Australia in 2005 and 2006, and at the Massachusetts Institute of Technology in 2006 and 2008. He was a postdoctoral fellow at the Laboratory for Information and Decision Systems of the Massachusetts Institute of Technology 2009 and 2010, holding postdoctoral fellowships of the F.R.S.-FNRS (Fund for Scientific Research) and of Belgian American Education Foundation. He was also resident scholar at the Center for Information and Systems Engineering (Boston University) in 2018-2019, holding a WBI.World excellence fellowship.

Doctor Hendrickx is the recipient of the 2008 EECI award for the best PhD thesis in Europe in the field of Embedded and Networked Control, and of the Alcatel-Lucent-Bell 2009 award for a PhD thesis on original new concepts or application in the domain of information or communication technologies. 
\end{IEEEbiography}

\begin{IEEEbiography}[{\includegraphics[width=1in,height=1.25in,clip,keepaspectratio]{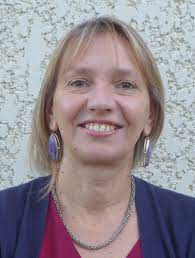}}]{Elena Panteley}
received the M.Sc. and Ph.D. degrees in applied mathematics from the State University of St. Petersburg, St. Petersburg, Russia, in 1986 and 1997, respectively. From 1986 to 1998, she held a research position with the Institute for Problem of Mechanical Engineering, Russian Academy of Science, St. Petersburg. Since 2004 she holds a tenure position as Senior Researcher of the French National Centre of Scientific Research (CNRS), at the Laboratoire de signaux et systèmes, France. She is also associate researcher of ITMO University, St Petersbourg Russia, since 2014. Her research interests include stability and control of nonlinear dynamical systems, network systems with applications to electromechanical and neuronal systems.
\end{IEEEbiography}

\begin{IEEEbiography}[{\includegraphics[width=1in,height=1.25in,clip,keepaspectratio]{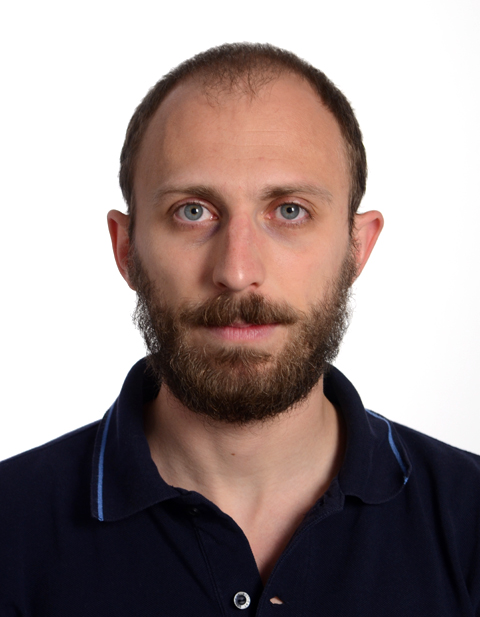}}]{Paolo Frasca}
(M'13, SM'18) received the Ph.D.\ degree in Mathematics for Engineering Sciences from Politecnico di Torino, Torino, Italy, in 2009. 
From 2013 to 2016, he was an Assistant Professor at the University of Twente in Enschede, the Netherlands. Since October 2016 he is a CNRS Researcher affiliated with GIPSA-lab, Grenoble, France. 
His research interests are in the theory of networks and control systems, with main applications to transportation and social networks. \end{IEEEbiography}

\end{document}